\documentclass[11pt]{article}

\usepackage{graphicx}
\usepackage{color}
\usepackage{amsmath}
\usepackage{amsthm}
\usepackage{enumerate}
\usepackage{fullpage}
\usepackage{subfigure}
\newcommand{\rev}[1]{\text{rev}(#1)}
\newcommand{\pend}[1]{\text{end}(#1)}
\newcommand{\pstart}[1]{\text{start}(#1)}
\newcommand{\rvec}{\boldsymbol{r}}
\newcommand{\cost}{\boldsymbol{c}}
\newcommand{\interior}{\text{int}}
\newenvironment{enumtight}
{
  \begin{enumerate}
    \setlength{\itemsep}{0pt}
    \setlength{\parsep}{0pt}
    \setlength{\topsep}{0pt}
    \setlength{\partopsep}{0pt}
    \setlength{\parskip}{0pt}
  }
  { \end{enumerate} }

\newcounter{tmp}

\newcounter{tmp2}

\newenvironment{smquotehelper}{\begin{list}{}{\setlength{\itemsep}{.25ex}\setlength{\labelsep}{0pt}
      \setlength{\topsep}{.25ex}\setlength{\leftmargin}{3em}
      \setlength{\parsep}{.1ex}}}{\end{list}}

\usepackage{bm}

\usepackage{amsmath, amssymb}
\newtheorem{lemma}{Lemma}[section] 
\newtheorem{definition}[lemma]{Definition}
\newtheorem{corollary}[lemma]{Corollary}
\newtheorem{theorem}[lemma]{Theorem}

\newcommand{\set}[1]{\ensuremath{\{#1\}}}

\renewcommand{\setminus}{-}

\newcommand{\cvec}{\bm{c}}

\newcommand{\dist}{\ensuremath{{dist}}}
\newcommand{\pathcat}{\ensuremath{\circ}}

\newcommand{\opt}{\ensuremath{{\rm OPT}}}

\newcommand{\BC}{{\cal B}^+}

\newcommand{\MG}{\ensuremath{\text{\it MG}}}

\newcommand{\PRG}{\ensuremath{\BC(\MG)}}

\newcommand{\CC}{{\cal B}^\div}

 \newcounter{xcount}

{\unskip\nobreak\hskip 1em plus 1fil\nobreak$\Box$
\parfillskip=0pt%
\endtrivlist}

\begin{document}
\title{The two-edge connectivity survivable-network design problem in planar
  graphs}
\author{Glencora Borradaile\\Oregon State University \and Philip
  Klein\\ Brown University}

\maketitle

\begin{abstract} 
  Consider the following problem: given a graph with edge costs and
  a subset $Q$ of vertices, find a minimum-cost subgraph in which
  there are two edge-disjoint paths connecting every pair of vertices
  in $Q$.  The problem is a failure-resilient analog of the Steiner
  tree problem arising, for example, in telecommunications
  applications.  We study a more general mixed-connectivity
  formulation, also employed in telecommunications optimization.
  Given a number (or {\em requirement}) $r(v)\in \set{0,1,2}$ for each
  vertex $v$ in the graph, find a minimum-cost subgraph in which
  there are $\min\set{r(u), r(v)}$ edge-disjoint $u$-to-$v$ paths for
  every pair $u,v$ of vertices.

  We address the problem in planar graphs, considering a popular
  relaxation in which the solution is allowed to use multiple copies
  of the input-graph edges (paying separately for each copy).  The
  problem is max SNP-hard in general graphs and strongly NP-hard in
  planar graphs.  We give the first polynomial-time approximation
  scheme in planar graphs.  The running time is $O(n \log n)$.

  Under the additional restriction that the requirements are only
  non-zero for vertices on the boundary of a single face of a planar
  graph, we give a polynomial-time algorithm to find the optimal
  solution.
\end{abstract}

\section{Introduction}

In the field of telecommunications network design, an important
requirement of networks is resilience to link
failures~\cite{TelecommunicationsHandbook}.  The goal of the
survivable network problem is to find a graph that provides multiple
routes between pairs of terminals.  In this work we address a problem
concerning edge-disjoint paths. 
 For a set $S$ of non-negative
integers, an instance of the {\em $S$-edge connectivity design} problem is a pair
$(G, \rvec)$ where $G = (V,E)$ is a 
undirected graph with edge costs $\cost\; : \; V \rightarrow \Re^+$ and {\em connectivity requirements} $\rvec\; : \; V \rightarrow S$.  The goal is to find a
minimum-cost subgraph of $G$ that, for each pair $u,v$ of
vertices, contains at least $\min\set{\rvec(u),\rvec(v)}$ edge-disjoint
$u$-to-$v$ paths.

In telecommunication-network design, failures are rare; for this
reason, there has been much research on {\em low-connectivity network design}
  problems, in which the maximum connectivity requirement is two.
 Resende and Pardalos~\cite{TelecommunicationsHandbook}
survey the literature, which includes heuristics, structural results,
polyhedral results, computational results using cutting planes, and
approximation algorithms.  This work focuses on $\set{0,1,2}$-edge
connectivity problems in planar graphs.

We consider the previously studied variant wherein the solution
subgraph is allowed to contain multiple copies of each edge of the
input graph (a {\em multi}-subgraph); the costs of the edges in the
solution are counted according to multiplicity.  For
$\set{0,1,2}$-connectivity, at most two copies of an edge are needed.
We call this the {\em relaxed} version of the problem and use the
term {\em strict} to refer to the version of the problem in which
multiple copies of edges of the input graph are disallowed.

A {\em polynomial-time approximation scheme} (PTAS) for an
optimization problem is an algorithm that, given a fixed constant
$\epsilon > 0$, runs in polynomial time and returns a solution within
$1 +\epsilon$ of optimal. The algorithm's running time need not be
polynomial in $\epsilon$.  The PTAS is {\em efficient} if the running
time is bounded by a polynomial whose degree is independent of
$\epsilon$.  In this paper, we focus on designing a PTAS for
$\set{0,1,2}$-edge connectivity.

\paragraph{Two edge-connected spanning subgraphs} 
A special case that has received much attention is the problem of
finding a minimum-cost subgraph of $G$ in which {\em every} pair of
vertices is two edge-connected.  Formally this is the strict
$\set{2}$-edge connectivity design problem.  This problem is
NP-hard~\cite{ET76} (even in planar graphs, by a reduction from
Hamiltonian cycle) and max-SNP hard~\cite{CL99} in general graphs. 
In general graphs, 
Frederickson and J\'{a}J\'{a}~\cite{FJ81} gave an approximation ratio of 3 which was later improved to 2
(and 1.5 for unit-cost graphs) by Khuller and Vishkin~\cite{KV94}.
In planar graphs, 
Berger et~al.~\cite{BCGZ05} gave a polynomial-time
approximation scheme (PTAS) for the relaxed $\set{1,2}$-edge connectivity design
problem and Berger and Grigni~\cite{BG07} gave a PTAS for the
strict $\set{2}$-edge connectivity design problem.  
Neither of these
algorithms is efficient; the degree of the polynomial bounding the
running time grows with $1/\epsilon$. 
For the relaxed version of spanning planar 2-edge connectivity, the techniques of
Klein~\cite{Klein08} can be used to 
obtain a linear-time approximation scheme.

\paragraph{Beyond spanning} 
When a vertex can be assigned a requirement of zero, edge-connectivity
design problems include the {\em Steiner tree problem}: given a graph
with edge costs and given a subset of vertices (called the {\em
  terminals}), find a minimum-cost connected subgraph that includes
all vertices in the subset.  More generally, we refer to any vertex with a non-zero connectivity requirement as a terminal.
For $\set{0,2}$-edge connectivity design problem, in general graphs,
Ravi~\cite{Ravi92} showed that Frederickson and J\'aJ\'a's approach
could be generalized to give a 3-approximation algorithm (in general
graphs).  Klein and Ravi~\cite{KR93} gave a 2-approximation for the
$\set{0,1,2}$-edge connectivity design problem.  (In fact, they solve
  the even more general version in which requirements $\rvec({u,v})$ are
  specified for {\em pairs} $u,v$ of vertices.)  This result was
generalized to connectivity requirements higher than two by Williamson
et~al.~\cite{WGMV93}, Goemans et~al.~\cite{GGPSTW94}, and
Jain~\cite{Jain01}.  These algorithms each handle the strict version of the problem.

\bigskip

In their recent paper on the spanning case~\cite{BG07}, Berger and
Grigni raise the question of whether there is a PTAS for the
$\set{0,2}$-edge connectivity design problem in planar graphs.  In this
paper, we answer that question in the affirmative for the
relaxed version.  The question in the case of the strict version is
still open.

\subsection{Summary of new results}

Our main result is a PTAS for the relaxed $\set{0,1,2}$-edge
connectivity problem in planar graphs:

\begin{theorem} \label{thm:main}
For any $\epsilon$, there is an $O(n \log n)$ algorithm that, given a
planar instance $(G, \rvec)$ of relaxed $\set{0,1,2}$-edge
connectivity, finds a solution whose cost is at most $1+\epsilon$
times optimal. 
\end{theorem}

This result builds on the work of Borradaile, Klein and
Mathieu~\cite{BKK07,BKM07,BKM09} which gives a PTAS for the Steiner
tree (i.e.~$\set{0,1}$-edge connectivity) problem.
  This is the first
PTAS for a non-spanning two-edge connectivity problem in planar graphs.

Additionally, we give an exact, polynomial-time algorithm for the
special case where are the vertices with non-zero requirement are on
the boundary of a common face:

\begin{theorem} \label{thm:2ec-exact} There is an $O(k^3n)$-time
  algorithm that finds an optimal solution to any planar instance $(G,
  \rvec)$ of 
relaxed $\set{0,1,2}$-edge-connectivity in which only $k$ vertices are assigned nonzero requirements and all of
 them are on the boundary of a single face.  For instances of relaxed
 $\set{0,2}$-edge connectivity (i.e. all requirements are 0 or 2), the
 algorithm runs in linear time.
\end{theorem}

\subsection{Organization}

We start by proving Theorem~\ref{thm:2ec-exact} in
Section~\ref{sec:exact}.  The proof of this result is less involved
and provides a good warm-up for the proof of Theorem~\ref{thm:main}.
The algorithm uses the linear-time shortest-path algorithm for planar
graphs~\cite{HKRS97} and a polynomial-time algorithm for the
equivalent boundary Steiner-tree problem~\cite{EMV87} as black boxes.

In order to prove Theorem~\ref{thm:main}, we need to review the
framework developed for the Steiner tree problem in planar graphs.  We
give an overview of this framework in Section~\ref{sec:ptas-framework} and show
how to use it to solve the relaxed $\set{0,1,2}$-edge connectivity
problem.  The correctness of the PTAS relies on a Structure Theorem
(Theorem~\ref{thm:structure}) which bounds the number of interactions of a solution between
different regions of the graph while paying only a small
relative penalty in cost.  We prove this Structure Theorem in
Section~\ref{sec:structure-theorem}.  The algorithm itself requires a
dynamic program; we give the details for this in Section~\ref{sec:dp}.

\section{Basics}

We consider graphs and multi-subgraphs.  A multi-subgraph is a
subgraph where edges may be included with multiplicity.  
In proving the Structure Theorem we will replace subgraphs of a solution with other subgraphs.  In doing so, two of the newly introduced subgraphs may share an edge.

For a subgraph $H$ of a graph $G$, we use $V(H)$ to denote the set of
vertices in $H$.  For a graph $G$ and set of edges $E$, $G / E$
denotes the graph obtained by contracting the edges $E$.

For a path $P$, $P[x,y]$ denotes the $x$-to-$y$ subpath of $P$ for
vertices $x$ and $y$ of $P$; $\pend{P}$ and $\pstart{P}$ denote the
first and last vertices of $P$; $\rev{P}$ denotes the reverse of path
$P$.  For paths $A$ and $B$, $A\pathcat B$ denotes the concatenation
of $A$ and $B$.  See Figure~\ref{fig:crossing} for an illustration of
the notion of paths crossing.  A cycle is non-self-crossing if every
pair of subpaths of the cycle do not cross.

We employ the usual definitions of planar embedded graphs.  For a face
$f$, the cycle of edges making up the boundary of $f$ is denoted
$\partial f$.  We assume the planar graph $G$ is connected and is
embedded in the plane, so there is a single infinite face, and we
denote its boundary by $\partial G$.

For a cycle $C$ in a planar embedded graph, $C[x,y]$ denotes an
$x$-to-$y$ path in $C$ for vertices $x$ and $y$ of $C$.  There are two
such paths and the choice between the two possibilities will be
disambiguated by always choosing the subpath in the clockwise
direction.  A cycle $C$ is said to {\em enclose} the faces that are
embedded inside it.  $C$ encloses an edge/vertex if the edge/vertex is
embedded inside it or on it.  In the former case, $C$ {\em strictly
  encloses} the edge/vertex.  For non-crossing $x$-to-$y$ paths $P$
and $Q$, $P$ is said to be left of $Q$ if $P \pathcat \rev{Q}$ is a
clockwise cycle.

\begin{figure}[ht]
  \centering
     \subfigure[]{\includegraphics[scale=0.8]{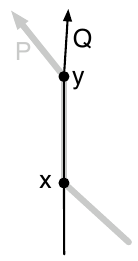}}
     \subfigure[]{\includegraphics[scale=0.8]{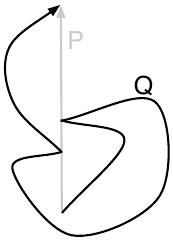}}
     \subfigure[]{\includegraphics[scale=1]{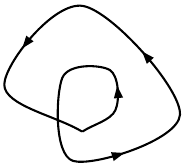}}
     \subfigure[]{\includegraphics[scale=1]{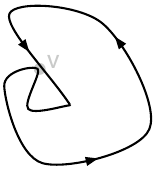}}
    \caption[Crossing paths and cycles]{(a) $P$ crosses $Q$. (b) $P$
      and $Q$ are noncrossing; $Q$ is left of $P$  (c) A self-crossing cycle. (d) A
      non-self-crossing cycle (non-self-crossing allows for repeated
      vertices, i.e.~$v$.)}
  \label{fig:crossing}
\end{figure}

We will use the following  as a subroutine:

\begin{theorem}\label{thm:emv87}{Erickson, Monma and Veinott~{\bf \cite{EMV87}}} 
  Let $G$ be a planar embedded graph with edge-costs and let $Q$ be a set of
  $k$ terminals that all lie on the boundary of a single face. Then
  there is an algorithm to find an minimum-cost Steiner tree of $G$
  spanning $Q$ in time $O(nk^3 + (n \log n)k^2)$ or $O(nk^3)$ time using the algorithm of~\cite{HKRS97}.
\end{theorem}

\subsection{Edge-connectivity basics}

Since we are only interested in connectivity up to and including
two-edge connectivity, we define the following: For a graph $H$ and
vertices $x,y$, let
$$c_H(x,y)=\min \set{2, \text{maximum number of edge-disjoint $x$-to-$y$
    paths in $H$}}.$$ For two multi-subgraphs $H$ and $H'$ of a common
graph $G$ and for a subset $S$ of the vertices of $G$, we say $H'$
{\em achieves the two-connectivity of $H$} for $S$ if $c_{H'}(x,y)
\geq c_H(x,y)$ for every $x,y\in S$.  We say $H'$ achieves the {\em
  boundary} two-connectivity of $H$ if it achieves the
two-connectivity of $H$ for $S = V(\partial G)$.

Several of the results in the paper build on observations of the
structural property of two-edge connected graphs.  The first is a well-known property:

\begin{lemma}[Transitivity] \label{lem:transitive} For any graph $H$, for vertices
  $u,v,w\in V(H)$, $c_H(u,w) \geq \min\{c_H(u,v), c_H(v,w)\}$
\end{lemma}

Note that in the following, we can replace ``strictly encloses no/strictly enclosed'' with
``strictly encloses all/strictly not enclosed'' without loss of generality (by viewing a face
enclosed by $C$ as the infinite face).

\begin{lemma}[Empty Cycle] \label{lem:empty-cycle} Let $H$ be a
  (multi-)subgraph of $G$ and let $C$ be a non-self-crossing cycle of
  $H$ that strictly encloses no terminals.  Let $H'$ be the subgraph
  of $H$ obtained by removing the edges of $H$ that are strictly
  enclosed by $C$.  Then $H'$ achieves the two-connectivity of $H$.
\end{lemma}

\begin{proof} See Figure~\ref{fig:cycle-2ec}(b).  Without loss of
  generality, view $C$ as a clockwise cycle.  Consider two terminals
   $x$ and $y$.  We show that there are $c_H(x,y)$
  edge-disjoint $x$-to-$y$ paths in $H$ that do not use edges strictly
  enclosed by $C$.  There are two nontrivial cases:

  \begin{description}
  \item[$c_H(x,y)=1:$] Let $P$ be an $x$-to-$y$ path in $H$.  If $P$
    intersects $C$, let $x_P$ be the first vertex of $P$ that is in
    $C$ and let $y_P$ be the last vertex of $P$ that is in $C$.  Let
    $P' = P[x,x_P] \pathcat C[x_P,y_P] \pathcat P[y_P,y]$.  If $P$
    does not intersect $C$, let $P' = P$.  $P'$ is an $x$-to-$y$ path
    in $H$ that has no edge strictly enclosed by $C$.
  \item[$c_H(x,y)=2:$] Let $P$ and $Q$ be edge-disjoint $x$-to-$y$
    paths in $H$.  If $Q$ does not intersect $C$, then $P'$ and $Q$
    are edge-disjoint paths, neither of which has an edge strictly
    enclosed by $C$ (where $P'$ is as defined above).  Suppose that
    both $P$ and $Q$ intersect $C$.  Define $x_Q$ and $y_Q$ as for
    $P$.  Suppose these vertices are ordered $x_P$, $x_Q$, $y_Q$,
    $y_P$ around $C$.  Then $P[x,x_P] \pathcat C[x_P,y_Q] \pathcat
    Q[y_Q,y]$ and $Q[x,x_Q] \pathcat \rev{C[y_P,x_Q]} \pathcat
    P[y_P,y]$ are edge disjoint $x$-to-$y$ paths that do not use any
    edges enclosed by $C$.  This case is illustrated in
    Figure~\ref{fig:cycle-2ec}; other cases (for other orderings of
    $\{x_P, x_Q, y_Q, y_P\}$ along $C$) follow similarly.
  \end{description}

\begin{figure}[ht]
  \centering
   {\includegraphics[scale=1]{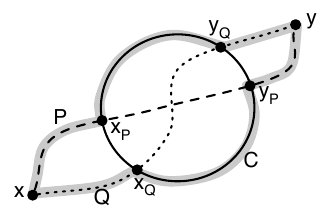}}
   \caption{An illustration of the proof of
     Lemma~\ref{lem:empty-cycle}: there are edge-disjoint $x$-to-$y$
     paths (grey) that do not use edges enclosed by $C$. }
  \label{fig:cycle-2ec}
\end{figure}

  \noindent We have shown that we can achieve the boundary
  two-connectivity of $H$ without using any edges strictly enclosed by a cycle of
  $H$.  The lemma follows.
 \end{proof}

\subsection{Vertex-connectivity basics} 

The observations in
this section do not involve planarity.  Although our results are for
edge connectivity, we use vertex connectivity in
Section~\ref{sec:structure-theorem} to simplify our proofs.

Vertices
$x$ and $y$ are {\em biconnected} (a.k.a.~two-vertex-connected) in a graph $H$ if $H$
contains two $x$-to-$y$ paths that do not share any internal vertices, or, equivalently, if there is a simple cycle in $H$ that contains both $x$ and $y$.  
For a subset $S$ of vertices of $H$, we say $H$ is {\em $S$-biconnected} if
for every pair $x,y$ of vertices of $S$, $H$ contains a simple cycle
through $x$ and $y$.  We refer to the vertices of $S$ as {\em terminals}.

\begin{lemma} \label{lem:His2VC}
A minimal $S$-biconnected graph is biconnected.
\end{lemma}

\begin{proof}
  Suppose $H$ has a cut-vertex $v$: $H=H_1 \cup H_2$ where $H_1 \cap
  H_2 = \{v\}$ and $H_i \ne \{v\}$ for $i = 1,2$.  If $H_1$ and $H_2$ both have terminals then $H$ does
  not biconnect every pair of terminals.  If, say, $H_2$ does
  not have a terminal then $H \setminus (H_2 \setminus \{v\})$ is a
  smaller subgraph that biconnects the terminals.
\end{proof}

For the next two proofs, we use the notion of an (open) ear decomposition.  An ear decomposition of a graph $G$ is a partition of the edges into a cycle $C$ and a sequence of paths $P_1, P_2, \ldots, P_k$ such that the endpoints of $P_i$ are in $C \cup( \bigcup_{j < i} P_j)$.  The ear decomposition is open if the endpoints of $P_i$ are distinct.  A graph is biconnected iff it has an open ear decomposition~\cite{Whitney32}.   Ear decompositions can be built greedily starting with any cycle.  

\begin{theorem} \label{thm:2vc-terminal}
Let $H$ be a minimal $S$-biconnected graph.
Every cycle in $H$ contains a vertex of $S$.
\end{theorem}

\begin{proof} 
  Assume for a contradiction that $H$ contains a cycle $C$ that does not contain any terminals.  By Lemma~\ref{lem:His2VC}, $H$ is biconnected and so has an open ear decomposition starting with $C$; let $C, P_1, P_2, \ldots, P_k$ be an open ear decomposition of $H$.  Define $\mathcal{E}_i$ to be the subgraph composed of $C$ and the first $i$ ears: $C \cup (\bigcup_{j \le i} P_j)$.
  We construct another open ear decomposition with one fewer ear $C', P_2', \ldots, P_k'$ of $H$ as follows (note there is no ear $P_1'$) and use $\mathcal{E}_i'$ to denote $C' \cup (\bigcup_{j \le i} P_j')$.  

Let $x$ and $y$ be the endpoints of $P_1$.  Let $C' = C[y,x] \circ P_1$.  Let $Q_1' = C[x,y]$ be the portion of $C$ that is not used in $C'$. We will maintain the invariant:
  \[Q_i' \text{ contains at least one edge and } Q_i' = \mathcal{E}_i\setminus \mathcal{E}_i'\] Clearly this invariant holds for $i = 1$.  For $i \ge 2$, we define $P_i'$ using $P_i$ and $Q_{i-1}'$.  Note that one or more of the endpoints of $P_i$ may be in $Q_{i-1}'$ and so $\mathcal{E}_{i-1}' \cup P_i$ is not necessarily a valid ear decomposition.  However, by the invariant, $P_i$'s endpoints are in $Q_{i-1}' \cup  \mathcal{E}_{i-1}'$, allowing us to define a new valid ear $P_i'$ by extending $P_i$ along $Q_{i-1}'$ to reach $\mathcal{E}_{i-1}'$ as follows:  $P_i'$ is the minimal path of $P_i \cup Q_{i-1}' \cup  \mathcal{E}_{i-1}'$ whose endpoints are in $\mathcal{E}_{i-1}'$
such that $P_i$ is a subpath of $P_i'$.  Define $Q_i' = Q_{i-1}' \setminus (P_i' \setminus P_i)$.  Since $P_i$ has distinct endpoints, $P_i'$ does not contain all the edges of $Q_{i-1}'$, thus maintaining the invariant.

  By construction, $C', P_2', \ldots, P_k'$ is an open ear decomposition of $H \setminus Q_k'$ and so $H \setminus Q_k'$ is biconnected.  Since $Q_k' \subset C$ and $C$ does not contain any terminals, $H \setminus Q_k'$ is $S$-biconnected and since $Q_k'$ contains at least one edge, $H \setminus Q_k'$ contradicts the minimality of $H$.
\end{proof}

\begin{theorem} \label{thm:cycle-and-path}
Let $H$ be a minimal $S$-biconnected graph.  For any cycle $C$ in $H$,
every $C$-to-$C$ path contains a vertex of $S$.
\end{theorem}

\begin{proof} 
  Let $C$ be any cycle.  By Theorem~\ref{thm:2vc-terminal}, $C$ contains a terminal.  We consider an ear decomposition $C, P_1, P_2, \ldots$  of $H$ built as follows.  Consider $s \in S$ not spanned by $C \cup P_1 \cup \cdots \cup P_{i-1}$. Then there are vertex-disjoint paths from $s$ to $C \cup P_1 \cup \cdots \cup P_{i-1}$ since $s$ and the terminal on, for example, $C$ are biconnected.  Let $P_i$ be the ear formed by these vertex-disjoint paths.  Observe that by this construction each ear $P_i$ contains a terminal for every $i$.

  Suppose every path in $\cup_{i \le k} P_i$ with two endpoints in $C$ strictly contains a vertex of $S$. We prove that this is then the case for $\cup_{i \le k+1} P_i$.  Since $P_{k+1}$ is an ear, its endpoints are in  $\cup_{i \le k} P_i$ and so any $C$-to-$C$ path that uses an edge of $P_{k+1}$ would have to contain the entirety of $P_{k+1}$; therefore $P_{k+1}$ cannot introduce a terminal-free path.
\end{proof}

\section{An exact algorithm for boundary $\{0,1,2\}$-edge
  connectivity} \label{sec:exact}  
Our algorithm for the boundary case of $\{0,1,2\}$-edge connectivity,
as formalized in Theorem~\ref{thm:2ec-exact} is based on the
observation that there is an optimal solution to the problem that is
the union of Steiner trees whose terminals are boundary vertices,
allowing us to employ the boundary-Steiner-tree algorithm of
Theorem~\ref{thm:emv87}.

When terminals are restricted to the boundary of the graph, no cycle
can strictly enclose a terminal.  By Lemma~\ref{lem:empty-cycle}, we
get:

\begin{corollary} \label{cor:empty-cycle} Let $H$ be a subgraph of $G$
  and let $H'$ be a minimal subgraph of $H$ that achieves the boundary
  two-connectivity of $H$.   Then in $H'$ every cycle $C$ strictly
  encloses no edges.
\end{corollary}

In the following we will assume that the boundary of the graph $G$ is
a simple cycle; that is, a vertex appears at most once along $\partial
G$.  Let us see why this is a safe assumption.  Suppose the boundary
of $G$ is not simple: there is a vertex $v$ that appears at least
twice along $\partial G$.  Partition $G$ into two graphs $G_1$ and
$G_2$ such that $G_1 \cap G_2 = v$, $v$ appears exactly once along $\partial G_1$ and
$E(\partial G) = E(\partial G_1) \cup E(\partial G_2)$.  Let $x$ be a
vertex of $\partial G_1$ and let $y$ be a vertex of $\partial G_2$.
Then $c_G(x,y) = \min \set{c_{G_1}(x,v), c_{G_2}(v,y)}$, allowing us
to define new connectivity requirements and solve the problem
separately for $G_1$ and $G_2$.

\begin{lemma} \label{lem:leftmost-paths-not-cross} Let $P$ and $Q$ be
 leftmost non-self-crossing $x_P$-to-$y_P$ and $x_Q$-to-$y_Q$ paths,
 respectively, where $x_P$, $y_P$, $x_Q$, and $y_Q$ are vertices in clockwise
 order on $\partial G$.  Then $P$ does not cross $Q$.
\end{lemma}

\begin{proof}
  For a contradiction, assume that $Q$ crosses $P$.  Refer to
  Figure~\ref{fig:leftmost-paths-not-cross}(a). Let $C$ (interior shaded)
  be the cycle $P \pathcat \rev{\partial G[x_P,y_P]}$.  $C$ strictly
  encloses neither $x_Q$ nor $y_Q$.  If $Q$ crosses $P$, there must be
  a subpath of $Q$ enclosed by $C$.  Let $x$ be the first vertex of
  $Q$ in $P$ and let $y$ be the last.  There are two cases:
 \begin{description}
 \item[$x \in P{[y,y_P]}:$] Refer to
   Figure~\ref{fig:leftmost-paths-not-cross}(a). In this case, $\rev{P[y,x]}$ is left of
     $Q[x,y]$ and so $Q[x_Q,x] \pathcat \rev{P[y,x]} \pathcat
     Q[y,y_Q]$ (grey path) is left of $Q$, contradicting the leftmostness
     of $Q$.
 \item[$x \in P{[x_P,y]}:$] Refer to
   Figure~\ref{fig:leftmost-paths-not-cross}(b). In this case,
   $Q[x_Q,x] \pathcat \rev{P[x_P,x]} \pathcat
   \partial G[x_P,x_Q]$ (shaded interior) is a cycle that strictly
   encloses $y$ and does not enclose $y_Q$.  Since $y$ is the last
   vertex of $Q$ on $P$, $Q$ must cross itself, a contradiction. \qedhere
 \end{description}
\end{proof}

\begin{figure}[ht]
  \centering
  \subfigure[]{\includegraphics{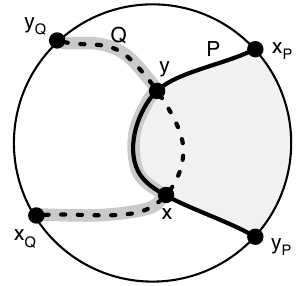}}
  \subfigure[]{\includegraphics{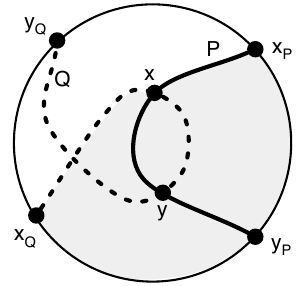}}
  \caption{Illustration of Lemma~\ref{lem:leftmost-paths-not-cross}:
    there exist leftmost paths that do not cross.}
  \label{fig:leftmost-paths-not-cross}
\end{figure}

\begin{lemma} \label{lem:clique-cycle} Let $H$ be a subgraph of $G$.
  Let $S$ be a subset of $V(\partial G)$ such that, for every $x,y\in
  S$, $c_H(x,y) = 2$.  Then there is a
  non-self-crossing cycle $C$ in $H$ such that $S \subseteq V(C)$ and
  the order that $C$ visits the vertices in $S$ is the same as their
  order along $\partial G$.
\end{lemma}

\begin{proof}
  Assume that the vertices of $S$ are in the clockwise order $s_0,
  s_1, \ldots, s_{k-1}$ along $\partial G$.

  Let $P_i$ be the leftmost non-self-crossing $s_{i-1}$-to-$s_i$ path
  in $H$, taking the indices modulo $k$.  Let $C = P_1 \pathcat P_2
  \pathcat \cdots \pathcat P_{k-1}$.  Certainly $C$ visits each of the
  vertices $s_0, s_1, \ldots$ in order.  By
  Lemma~\ref{lem:leftmost-paths-not-cross}, $P_i$ does not cross
  $P_j$ for all $i \neq j$.  Therefore, $C$ is non-self-crossing, proving the lemma.
\end{proof}

We now give an algorithm for the following problem: given a planar
graph $G$ with edge costs
and an assignment $\rvec$ of requirements such that $r(v)>0$ only
for vertices $v$ of $\partial G$, find a minimum-cost multi-subgraph $H$ of $G$ that
satisfies the requirements (i.e.~such that there are at least $\min\set{r(x),r(y)}$
edge-disjoint $x$-to-$y$ paths in $H$).

\begin{center}
  \fbox{
    \begin{minipage}[h]{.9\linewidth}
      \noindent {\sc Boundary2EC}$(G,r)$
      \begin{enumerate}
        \setlength{\itemsep}{0pt}
        \setlength{\parsep}{0pt}
        \setlength{\topsep}{0pt}
        \setlength{\partopsep}{0pt}
        \setlength{\parskip}{0pt}
      \item Let $q_1, q_2, \ldots$ be the cyclic ordering of
        vertices $\set{v \in V(\partial G) :\ \rvec(v) = 2}$.
      \item For $i = 1, \ldots$, let $X_i = \set{q_i} \cup \set{v \in V(\partial
          G[q_i,q_{i+1}])\ :\ \rvec(v) = 1} \cup \set{q_{i+1}}$.
      \item For $i = 1, \ldots$, let $T_i$ be the minimum-cost Steiner
        tree spanning $X_i$.
      \item Return the disjoint union $\cup_i T_i$.
      \end{enumerate}
    \end{minipage}
  }
\end{center}

We show that {\sc Boundary2EC} correctly finds the minimum-cost multi-subgraph of
$G$ satisfying the requirements.  Let $\opt$ denote an optimal
solution.  By Lemma~\ref{lem:clique-cycle}, $\opt$ contains a
non-self-crossing cycle $C$ that visits $q_1, q_2, \ldots$ (as defined
in {\sc Boundary2EC}).  By Corollary~\ref{cor:empty-cycle}, $C$
strictly encloses no edges of $\opt$.  Let $P_i$ be the leftmost
$q_i$-to-$q_{i+1}$ path in $C$. The vertices in $X_i$ are connected in
$\opt$, by the input requirements.  Let $S_i$ be the subgraph of
$\opt$ that connects $X_i$.  This subgraph is enclosed by $\partial
G[q_i,q_{i+1}]\pathcat C[q_i,q_{i+1}]$.  Replacing $S_i$ by $T_i$
achieves the same connectivity among vertices $v$ with $\rvec(v)>0$ without
increasing the cost.

We will use the following lemma to give an efficient implementation of
{\sc Boundary2EC}.

\begin{lemma} \label{lem:shortest-paths-enclosed} Let $a$, $b$ and $c$
  be vertices ordered along the clockwise boundary $\partial G$ of a
  planar graph $G$.  Let $T_a$ be the shortest-path tree rooted at $a$
  (using edge costs for lengths).  Then for any set of terminals $Q$
  in $\partial G[b,c]$, there is a minimum-cost Steiner tree
  connecting them that enclosed by the cycle $\partial G[b,c] \pathcat
  T_a[c,b]$.
\end{lemma}

\begin{proof} Refer to Figure~\ref{fig:tree-enclosed-by-cycle}.  Let
  $C = \partial G[b,c] \pathcat T_a[c,b]$.  Let $T$ be a minimum-cost
  Steiner tree in $G$ connecting $Q$.  Suppose some part of $T$ is not
  enclosed by $C$.  Let $T'$ be a maximal subtree of $T$ not enclosed
  by $C$.  The leaves of $T'$ are on $T_a[c,b]$.  Let $P$ be the
  minimum subpath of $T_a[b,c]$ that spans these leaves.  Let $P'$ be
  the $\pstart{P}$-to-$\pend{P}$ path in $T'$.  See
  Figure~\ref{fig:tree-enclosed-by-cycle}.

  We consider the case when $\pstart{P'}$ is a vertex of $T_a[a,b]$
  and $\pend{P'}$ is a vertex of $T_a[a,c]$ (the other cases, when
  $\pstart{P'}$ and $\pend{P'}$ are either both vertices of $T_a[a,b]$
  or both vertices of $T_b[a,c]$, are simpler).  Then $P'$ must cross
  $T_a[a,x]$ where $x$ is the last vertex common to $T_a[a,b]$ and
  $T_a[a,c]$ (i.e. the lowest common ancestor in $T_a$ of $b$ and $c$).  Let $y$ be a vertex of $P' \cap T_a[a,x]$.  Since $T_a$
  is a shortest-path tree in an undirected path, every subpath of
  $T_a[a,z]$ and $T_a[z,a]$, for any vertex $z$, is a shortest path.
  We have that:
  \begin{align*}
    \cost(P') &=\cost(P'[\pstart{P'},y]) + \cost(P'[y,\pend{P'}]) \\
    &\geq \cost(T_a[\pstart{P'},y]) + \cost(T_a[y,\pend{P'}]) \\
    &\geq \cost(T_a[\pstart{P'},\pend{P'}]) \\
    &\geq \cost(P)
  \end{align*}
  Let $\widehat T = T \setminus T' \cup P$.  By construction,
  $\widehat T$ spans $Q$.  Using that $\cost(P') \geq \cost(P)$, we have that
  $\cost(\widehat T) = \cost(T) - \cost(T') + \cost(P) \leq \cost(T) - \cost(T') + \cost(P') \leq
  \cost(T)$ since $P'$ is a subpath of $T'$.

  Repeating this process for every subtree of $T$ not enclosed by $C$
  results in a tree enclosed by $C$ spanning $Q$ that is no longer
  than $T$.
\end{proof}

\begin{figure}[ht]
  \centering
  \includegraphics{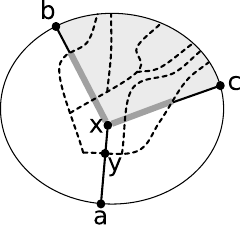}
  \caption[Existence of a tree enclosed by a cycle]{There is a tree
    $\widehat T$  that is just as cheap as $T$ (dotted) and spans
    the terminals between $b$ and $c$ but is enclosed by $C$ (whose interior is
    shaded).  $\widehat T$ is composed of the portion of $T$ enclosed
    by $C$ plus $P$, the thick grey path.}
  \label{fig:tree-enclosed-by-cycle}
\end{figure}

We describe an $O(k^3n)$-time implementation of {\sc Boundary2EC}
(where $k$ is the number of terminals). Compute a shortest-path tree
$T$ rooted at terminal $q_1$ in linear time.  For each $i$, consider
the graph $G_i$ enclosed by $C_i=\partial G[q_i,q_{i+1}] \pathcat
T[q_{i+1},q_i]$.  Compute the minimum Steiner tree spanning $X_i$ in
$G_i$.  By Lemma~\ref{lem:shortest-paths-enclosed}, $T_i$ has the same
cost as the minimum spanning tree spanning $X_i$ in $G$.  Since each
edge of $G$ appears in at most two subgraphs $G_i$ and $G_j$, the
trees $T_i$ can be computed in $O(k^3n)$ time
(by Theorem~\ref{thm:emv87}). 
 
Note: if the requirements are such that $\rvec(v) \in \set{0,2}$ for every
vertex $v$ on the boundary of $G$, then the sets $X_i$ have
cardinality 2.  Instead of computing Steiner trees in Step~3, we need
only compute shortest paths.  The running time for this special case
is therefore linear.

This completes the proof of
Theorem~\ref{thm:2ec-exact}.

\section{A PTAS framework for connectivity problems in planar graphs}
\label{sec:ptas-framework} 
In this section, we review the approach used in~\cite{BKM09} to give a
PTAS {\em framework} for the Steiner tree problem in planar graphs.  While the contents of this section are largely a summary of
the framework, we generalize where necessary for the survivable
network design problem, but refer the reader to the original paper~\cite{BKM09} for
proofs and construction details that are not unique to the focus of
this article.

Herein, denote the set of terminals by $Q$.  $\opt$ denotes an optimal
solution to the survivable network design problem.  We overload this
notation to also represent the cost of the optimal solution.

\subsection{Mortar graph and bricks}

The framework relies on an algorithm for finding a subgraph $\MG$ of
$G$, called the {\em mortar graph}~\cite{BKM09}.  The mortar graph
spans $Q$ and has total cost no more than $9\epsilon^{-1}$ times the
cost of a minimum Steiner tree in $G$ spanning $Q$
(Lemma~6.9 of~\cite{BKM09}).  Since a solution to the survivable network
problem necessarily spans $Q$, the mortar graph has cost at most
\begin{equation}
  9\epsilon^{-1} \cdot \opt.\label{eq:MG-length}
\end{equation}
The algorithm for computing $\MG$ first computes a 2-approximate
Steiner tree ~\cite{Mehlhorn88,Widmayer86,WWW86} and then augments
this subgraph with short paths.  The resulting graph is a grid-like
subgraph (the bold edges in Figure~\ref{fig:pcg}(a)) many of whose subpaths
are $\epsilon$-short:
\begin{definition}
  A path $P$ in a graph $G$ is $\epsilon$-short if for every pair of
  vertices $x$ and $y$ on $P$,
  $$\dist_P(x,y) \leq (1+\epsilon)\dist_G(x,y).$$
That is, the distance from $x$ to $y$ along $P$
  is at most $(1+\epsilon)$ times the distance from $x$ to $y$ in $G$.
\end{definition}

For each face $f$ of the mortar graph, the subgraph of $G$
enclosed by that face (including the edges and vertices of the face
boundary) is called a {\em brick}
(Figure~\ref{fig:pcg}(b)), and the brick's boundary is defined to be
$f$.    The boundary of a brick $B$ is written
$\partial B$.  The {\em interior of $B$} is defined to be the subgraph of
edges of $B$ not belonging to $\partial B$.  The interior of $B$ is
written $\interior(B)$.  

Bricks satisfy the following:
\begin{lemma}[Lemma~6.10~\cite{BKM09}]\label{lemma:brick-properties}
  The boundary of a brick $B$, in counterclockwise order, is the
  concatenation of four paths $W_B, S_B, E_B, N_B$ (west, south,
  east, north) such that:
  \begin{enumerate}
  \item The set of edges $B \setminus \partial B$ is nonempty.
  \item Every vertex of $Q \cap B$ is in $N_B$ or in $S_B$.
  \item $N_B$ is 0-short in $B$, and every proper subpath of $S_B$ is
    $\epsilon$-short in $B$.
  \item There exists a number $t\leq \kappa(\epsilon)$ and vertices
    $s_0, s_1, s_2, \ldots, s_t$ ordered from west to east along $S_B$
    such that, for any vertex $x$ of $S_B[s_i,s_{i+1})$, the distance
    from $x$ to $s_i$ along $S_B$ is less than $\epsilon$ times the
    distance from $x$ to $N_B$ in $B$: $\dist_{S_B}(x,s_i) <
    \epsilon\, \dist_{B}(x,N_B)$.
  \end{enumerate}
\end{lemma}
The number $\kappa(\epsilon)$ is given by:
\begin{equation}
  \label{eq:kappa}
  \kappa(\epsilon) = 4\epsilon^{−2}(1 + \epsilon^{−1})
\end{equation}

The mortar graph has some additional properties.  Let $B_1$ be a
brick, and suppose $B_1$'s eastern boundary $E_{B_1}$ contains at
least one edge.  Then there is another brick $B_2$ whose western
boundary $W_{B_2}$ exactly coincides with $E_{B_1}$.  Similarly, if
$B_2$ is a brick whose western boundary contains at least one edge
then there is a brick $B_1$ whose eastern boundary coincides with
$B_2$'s western boundary.

The paths forming eastern and western boundaries of
bricks are called {\em supercolumns}.

\begin{lemma}[Lemma~6.6~\cite{BKM09}]\label{lem:col-length}
  The sum of the costs of the edges in supercolumns is at most
  $\epsilon\, \opt$.
\end{lemma}

The mortar graph and the bricks are building blocks of the structural
properties required for designing an approximation scheme.
Borradaile, Klein and Mathieu demonstrated that there is a near-optimal
Steiner tree whose interaction with the mortar graph is ``simple''~\cite{BKM09}.
We prove a similar theorem in Section~\ref{sec:structure-theorem}.  In
order to formalize the notion of ``simple'', we select a subset of
vertices on the boundary of each brick, called {\em portals}, and define
a {\em portal-connected graph}.

\subsection{Portals and simple connections}

We define a subset of $\theta$ evenly spaced vertices along the
boundary of every brick.  The value of $\theta$ depends polynomially
on the precision $\epsilon$ and on $\alpha$, a parameter that represents how complex the solution can be within a single brick ($\alpha$ will be defined precisely in Equation~(\ref{eq:alpha}))
\begin{equation}
  \label{eq:theta}
  \theta(\epsilon)\mbox{ is } O(\epsilon^{-2} \alpha(\epsilon))
\end{equation}
The portals are selected to satisfy the following:
\begin{lemma}[Lemma~7.1~\cite{BKM09}]\label{lem:portal-distance}
  For any vertex $x$ on $\partial B$, there is a portal $y$ such that
  the cost of the $x$-to-$y$ subpath of $\partial B$ is at most $1/\theta$
  times the cost of $\partial B$.
\end{lemma}

Recall that, for each face $f$ of the mortar graph $MG$, there is a
corresponding brick $B$, and that $B$ includes the vertices and edges
comprising the boundary of $f$.  The next graph construction starts
with the disjoint union of the mortar graph $MG$ with all the bricks.
Each edge $e$ of the mortar graph is represented by three edges in the
disjoint union: one in $MG$ (the {\em mortar-graph copy} of $e$) and
one on each of two brick boundaries (the {\em brick copies} of $e$).
Similarly, each vertex on the boundary of a brick occurs several times
in the disjoint union. 

 The {\em portal-connected graph}, denoted
$\BC(MG)$, is obtained from the disjoint union as follows: a copy of each brick
$B$ of $MG$ is embedded in the interior of the corresponding face of $MG$, and each
portal $p$ of $B$ is connected by an  artificial edge to the
corresponding vertex of $MG$.  The construction is illustrated in
Figure~\ref{fig:pcg}(c).  The artificial edges are called {\em portal
  edges}, and are assigned zero cost.

Noting that each vertex $v$ on the boundary of a brick occurs
several times in $\BC(MG)$, we identify the original vertex $v$ of $G$
with that duplicate in $\BC(MG)$ that belongs to $MG$.  In particular,
each terminal (vertex in $Q$) is considered to appear exactly once in
$\BC(MG)$, namely in $MG$.  Thus the original instance gives rise to
an instance in $\BC(MG)$: the goal is to compute the optimal solution
w.r.t.\ the terminals on $MG$ in $\BC(MG)$ and then map the edges of
this solution to $G$.  Since $G$ can be
obtained from $\BC(MG)$ by contracting portal edges and identifying
all duplicates of each edge and all duplicates of each vertex, we infer:

\begin{lemma} \label{lem:feasibility}
Let $H$ be a subgraph of $\BC(MG)$ that, for each pair
  of terminals $u,v$, contains at least $\min\set{r(u),r(v)}$
  edge-disjoint $u$-to-$v$ paths.  Then the subgraph of $G$ consisting
  of edges of $H$ that are in $G$ has the same property.
\end{lemma}

The graph $\CC(MG)$, which we call the {\em brick-contracted graph}, is obtained by contracting each brick in $\BC(MG)$ to
a single vertex, called a {\em brick vertex}, as illustrated in
Figure~\ref{fig:pcg}(d).  This graph will be used in designing the
dynamic program in Section~\ref{sec:dp}.

\begin{figure}[ht]
  \centering
  \subfigure[]{\includegraphics[scale=1.6]{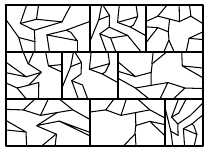}}
  \subfigure[]{\includegraphics[scale=1.6]{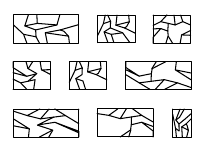}}
  \subfigure[]{\includegraphics[scale=1.6]{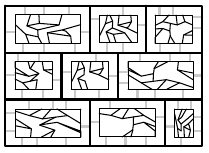}}
  \subfigure[]{\includegraphics[scale=1.6]{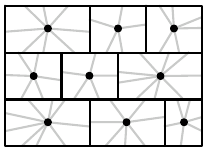}}
  \caption{(a) The mortar graph in bold, (b) the set of bricks,
   (c) the portal-connected graph $\BC(MG)$, and (d) the brick-contracted
    graph $\CC(MG)$.}
  \label{fig:pcg}
\end{figure}

\subsection{Structure Theorem}

Lemma~\ref{lem:feasibility} implies that, to find an approximately
optimal solution in $G$, it suffices to find a solution in $\BC(MG)$
whose cost is not much more than the cost of the optimal solution
in $G$.  The following theorem, which we prove in
Section~\ref{sec:structure-theorem}, suggests that this goal is
achievable.  An equivalent theorem was proven for the Steiner tree
problem~\cite{BKM09}.

\begin{theorem}[Structure Theorem]\label{thm:structure}  
  For any $\epsilon>0$ and any planar instance $(G, \rvec)$ of the
  $\{0,1,2\}$-edge connectivity problem, there exists a feasible
  solution $S$ to the corresponding instance $(\BC(MG), \rvec)$ such
  that 
  \begin{itemize}
  \item the cost of $S$ is at most $(1+c\epsilon)\opt$ where $c$ is
    an absolute 
    constant, and
  \item the intersection of $S$ with any brick $B$ is the union of a
    set of non-crossing trees whose leaves
    are portals.
  \end{itemize}
\end{theorem}

\subsection{Approximation scheme} \label{sec:ptas}

We assume that the input graph $G$ has degree at most three.  This can
be achieved using a well-known embedding-preserving transformation in
which each vertex of degree $d>3$ is replaced with a cycle of $d$
degree-three vertices, as shown in Figure~\ref{fig:degree-cycle}.
Making this assumption simplifies the dynamic program using for Step~5
below.

\begin{figure} 
\centerline{\includegraphics{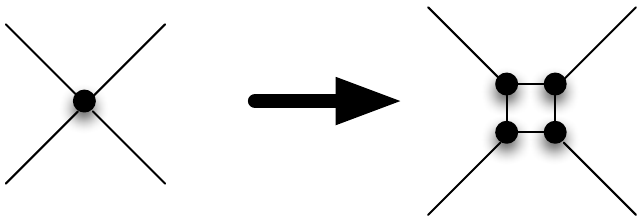}}
\caption{Each vertex of degree greater than $3$ is replaced with a
  cycle of degree-three vertices}
\label{fig:degree-cycle}\end{figure}

The approximation
scheme consists of the following steps.
\begin{enumerate}[Step 1:]
\item Find the mortar graph $\MG$.
\item Decompose $\MG$ into ``parcels'', subgraphs with the
  following properties: \vspace{-2.5mm}
  \begin{enumerate}[(a)]
  \item The parcels partition the 
    faces of $\MG$.  Since each edge of $\MG$ belongs to the
    boundaries of exactly two faces, it follows that each edge belongs
    to at most two parcels.
  \item The cost of all boundary edges (those edges belonging to two
    parcels) is at most $\frac{1}{\eta} \cost(\MG)$.  We choose $\eta$ so
    that this bound is $\frac{\epsilon}{2}\cost(OPT)$:
    \begin{equation}\label{eq:eta}
      \eta=\eta (\epsilon)=\lceil 20\epsilon^{-2}\rceil
    \end{equation}
  \item The planar dual of each parcel has a spanning tree of depth at
    most $\eta+1$.
  \end{enumerate}
  Each parcel $P$ corresponds to a subgraph of $G$, namely the
  subgraph consisting of the bricks corresponding to the faces making
  up $P$.  Let us refer to this subgraph as the {\em filled-in} version of $P$.
\item Select a set of ``artificial'' terminals on the boundaries of
  parcels to achieve the following:
\begin{itemize}
\item for each filled-in parcel, there is a solution that is feasible with
  respect to original and artificial terminals whose cost
  is at most that of the parcel's boundary plus the cost of the intersection
  of OPT with the filled-in parcel, and
\item  the union over all parcels of  such feasible solutions is a
  feasible solution for the original graph.
\end{itemize}
\item Designate portals on the boundary of each brick.
\item For each filled-in parcel $P$, find a optimal solution in the
  portal-connected graph, $\BC(P)$. Output the union of these
  solutions.
\end{enumerate}
Step~1 can be carried out in $O(n \log n)$ time~\cite{BKM09}.  Step~2 can be done
in linear time via breadth-first search in the
planar dual of $\MG$, and then applying a ``shifting'' technique in
the tradition of Baker~\cite{Baker94}.  Step~3 uses the fact that each
parcel's boundary consists of edge-disjoint, noncrossing cycles.  If
such a cycle separates terminals, a vertex $v$ on the cycle is
designated an artificial terminal.  We set $r(v) = 2$ if the cycle
separates terminals with requirement 2 and $r(v) = 1$ otherwise.
Under this condition, any feasible solution for the original graph
must cross the cycle; by adding the edges of the cycle, we get a
feasible solution that also spans the artificial terminal.  Step~3 can
be trivially implemented in linear time.  Step~5 is achieved in linear time
using dynamic programming (Section~\ref{sec:dp}).

\section{Proof of the Structure Theorem} \label{sec:structure-theorem}

We are now ready to prove the Structure Theorem for $\{0,1,2\}$-edge
connectivity, Theorem~\ref{thm:structure}.  
In order to formalize the notion of
connectivity across the boundary of a brick, we use the following
definition:

\begin{definition}[Joining vertex]
  Let $H$ be a subgraph of  $G$ and $P$ be a subpath of $\partial G$.
  A {\em joining vertex} of $H$ with $P$ is a vertex of $P$ that is
  the endpoint of an edge of $H \setminus P$.
\end{definition}

We will use the following structural lemmas in simplifying \opt.  The
first two were used in proving a Structure Theorem for the Steiner
tree PTAS~\cite{BKM09}; in these, $T$ is a tree and $P$ is an
$\epsilon$-short path on the boundary of the graph in which $T$ and
$P$ are embedded.  The third is in fact a generalization of the second
lemma that we require for maintaining two connectivity.

\begin{lemma}[Simplifying a tree with one root,
  Lemma~10.4~\cite{BKM09}] \label{lem:fib-tree} 

  Let $r$ be a vertex of $T$.  There is another tree $\widehat T$ that
  spans $r$ and the vertices of $T \cap P$ such that $\cost(\widehat T)
  \leq (1+4 \cdot \epsilon)\cost(T)$ and $\widehat T$ has at most $11
  \cdot \epsilon^{-1.45}$ joining vertices with $P$.
\end{lemma}

\begin{lemma}[Simplifying a tree with two roots,
  Lemma~10.6~\cite{BKM09}] \label{lem:short-runs} Let $p$ and $q$ be two
  vertices of $T$.  There is another tree $\widehat{T}$ that spans $p$
  and $q$ and the vertices of $T \cap P$ such that $\cost(\widehat{T})
  \leq (1+c_1\epsilon)\cost(T)$ and $\widehat T$ has at most $c_2 \cdot
  \epsilon^{-2.5}$ joining vertices with $P$, where $c_1$ and $c_2$
  are constants.
\end{lemma}

\begin{lemma}\label{lem:forest-cycle}
  Let $\mathcal F$ be a set of non-crossing trees whose leaves are vertices
  of $\epsilon$-short boundary paths $P$ and $Q$ and such that each
  tree in the forest has leaves on both these paths.  There is a cycle
  or empty set $\widehat C$, a
  set $\widehat{\mathcal F}$ of trees, and a mapping $\phi: \mathcal F
  \longrightarrow \widehat{\mathcal F}\cup \set{\widehat C}$ with the following properties
  \begin{itemize}
  \item For every tree $T$ in $\mathcal F$, $\phi(T)$ spans $T$'s leaves.
  \item For two trees $T_1$ and $T_2$ in $F$, if $\phi(T_i)\neq
    \widehat C$ for at least one of $i=1,2$ then $\phi(T_1)$ and
    $\phi(T_2)$ are edge-disjoint (taking into account edge multiplicities).
  \item The subgraph $\bigcup \widehat {\mathcal F} \cup \set{\widehat C}$ has $o(\epsilon^{-2.5})$ joining
    vertices with $P \cup Q$.
  \item $\cost(\widehat C)+ \sum \set{\cost(T)\ :\ T \in \widehat
      {\mathcal F}} \le 3\cost(Q) + (1+d \cdot
    \epsilon)\sum \set{\cost(T)\ : T\in \mathcal F}$ where $d$ is an absolute constant.
  \end{itemize}
\end{lemma}

\begin{proof}
  View the embedding of the boundary such that $P$ is on top and $Q$
  is at the bottom.  Let $T_1, \ldots, T_k$ be the trees of $F$
  ordered according the order of their leaves from left to right.

\bigskip

There are two cases.

Case 1) $k > 1/\epsilon$.  In this case, we reduce the number of trees
by incorporating a cycle $\widehat C$.  Let $a$
  be the smallest index such that $\cost(T_a) \le \epsilon \cost(F)$ and let
  $b$ be the largest index such that $\cost(T_b) \le \epsilon \cost(F)$.  We
  will replace trees $T_a, T_{a+1}, \ldots, T_b$ with a cycle.
  Let $Q'$ be the minimal subpath of $Q$
  that spans the leaves of $\bigcup_{i = a}^b T_i$ on $Q$.  We likewise
  define $P'$.  Let $L$ be the leftmost $Q$-to-$P$ path in $T_a$ and
  let $R$ be the rightmost $Q$-to-$P$ path in $T_b$.  Since $P$ is
  $\epsilon$-short, 
  \begin{equation}
    \cost(P') \le (1+\epsilon)\cost(L \cup Q' \cup R).\label{eq:1}
  \end{equation}
  To obtain $\widehat{F}$ from $F$, we  replace the trees $T_a, \ldots, T_b$ with the cycle
  $\widehat C = P'\cup L \cup Q' \cup R$ and set $\phi(T_a), \ldots,
  \phi(T_b)$ to $\widehat C$.  By construction $\widehat C$ spans the
  leaves of $\cup_{i = a}^b T_i$.

Case 2)  $k \leq 1/\epsilon$.  In this case,  the number of trees is
already bounded.    We set $a=2, b=1$ so as to not eliminate any
trees, and we set $\widehat C$ to be the empty set.

\bigskip

In both cases,
  for each remaining tree $T_i$ ($i \ne a,a+1, \ldots, b$) we do the
  following.  Let $T_i'$ be a minimal subtree of $T_i$ that spans all
  the leaves of $T_i$ on $P$ and exactly one vertex $r$ of $Q$.  Let
  $Q_i$ be the minimal subpath of $Q$ that spans the leaves of $T_i$
  on $Q$. We replace $T_i$ with the tree $\widehat T_i$ that
  is the union of $Q_i'$ and the tree guaranteed by
  Lemma~\ref{lem:fib-tree} for tree $T_i'$ with root $r$ and
  $\epsilon$-short path $P$.  By construction $\widehat T_i$ spans
  the leaves of $T_i$.  We set $\phi(T_i) = \widehat T_i$ for $i \ne a, \ldots, b$. 

$\widehat C$  has at most four joining vertices with $P \cup Q$.  Each tree
  $\widehat T_i$ has one joining vertex with $Q$ and, by
  Lemma~\ref{lem:fib-tree}, $o(\epsilon^{-1.5})$ joining vertices with
  $P$.  By the choice of $a$ and $b$, there are at most $2/\epsilon$
  of the trees in the second part of the construction.  This yields
  the bound on joining vertices.

  The total cost of the replacement cycle is:
  \begin{eqnarray*}
  \cost(\widehat C) & \le & \cost(P')+\cost(L)+\cost(Q')+\cost(R)\\
  & \le & (2+\epsilon)(\cost(L)+\cost(Q')+\cost(R)) \qquad \mbox{by Equation~\eqref{eq:1}}\\
  & \le & (2+\epsilon)(\cost(T_a)+\cost(Q')+\cost(T_b)) \qquad
  \mbox{since $L$ and $R$ are paths in $T_a$ and $T_b$}\\
  & \le & (2+\epsilon)(2\epsilon \cost(F)+\cost(Q')) \qquad \mbox{by the
    choice of $a$ and $b$}\\
  & \le & (4\epsilon+2\epsilon^2)\cost(F)+(2+\epsilon)\cost(Q')
  \end{eqnarray*}
  The total cost of the replacement trees is:
  \begin{eqnarray*}
  \sum_{i = 1, \ldots, a-1, b+1 \ldots k} \cost(\widehat T_i) & \le & \sum_{i = 1, \ldots, a-1, b+1 \ldots k}
  \cost(Q_i')+(1+4\epsilon)\cost(T_i') \qquad \mbox{by Lemma~\ref{lem:fib-tree}}\\
  & \le & \sum_{i = 1, \ldots, a-1, b+1 \ldots k}
  \cost(Q_i')+(1+4\epsilon)\cost(T_i) \qquad \mbox{since $T_i'$ is a subtree
    of $T_i$}
  \end{eqnarray*}
  By the ordering of the trees and the fact that they are
  non-crossing, $Q'$ and the $Q_i'$'s are disjoint.  Combining the
  above gives the bound on cost.
\end{proof}

\subsection{Construction of a new solution}

We start with a brief overview of the steps used to prove the
structure theorem.  We start with an edge multiset forming an
optimal solution, $\opt$.  Each 
step modifies either the input graph $G$ or a subgraph thereof while
simultaneously modifies the solution.  The
graphs and edge multisets resulting from these steps are denoted by
subscripts.  Details are given in subsequent sections.

\begin{description}
\item[Augment]   We add two copies of each supercolumn, obtaining $G_A$
  and $\opt_A$.  We consider the two copies to
  be interior to the two adjacent bricks. This step allows us, in the {\em restructure} step, to
  concern ourselves only with connectivity between the north and south boundaries of a
  brick.  

\item[Cleave] {\em Cleaving} a vertex refers to splitting it into two
  vertices and adding an artificial edge between the two vertices.
  In the {\em cleave} step, 
  we modify $G_A$ (and so in turn modify $MG_A$ and
  $\opt_A$) to create $G_C$ (and $MG_C$ and $\opt_C$) by cleaving
  certain vertices while maintaining a planar embedding.  Let $J_C$ be
  the set of artificial edges introduced.  Note that $X_C / J_C =  X_A$. 
  The artificial edges are assigned zero cost so the metric
  between vertices is preserved.  The artificial edges are added to
  the solution, possibly in multiplicity, so connectivity is preserved.

\item[Flatten] In this step, for each brick $B$, we consider the
  intersection  of the solution with $\interior(B)$; we replace some of the
  connected components of the intersection with subpaths of the
  boundary of $B$.  We denote the resulting solution by $\opt_F$.

\item[Map] We map the edges of $\opt_F$ to
  $\mathcal{B}^+(MG)$ creating $\opt_M$. This step temporarily
  disconnects the solution.

\item[Restructure] In this step, we modify the solution $\opt_M$.  For
  each brick $B$ in $MG_C$, the part of $\opt_M$ strictly interior to
  $B$ is replaced with another subgraph that has few joining vertices
  with $\partial B$.  We denote the resulting solution by $\opt_S$.

\item[Rejoin] In order to re-establish connections broken in the
  {\em Map} step, we add edges to $\opt_S$.  Next, we contract
  the artificial edges added in the {\em cleave} step.  We denote the
  resulting solution by  $\widehat{OPT}$.
\end{description}

Note that the solutions
$\opt$, $\opt_A$, and so on are multisets; an edge can occur more than
once.  We now describe these steps in greater detail.

\subsubsection{Augment} 

Recall that a supercolumn is the eastern boundary of one brick and the
western boundary of another, and that the sum of costs of all
supercolumns is small.  In the {\em Augment} step, for each supercolumn $P$,
we modify the graph as
shown in Figure~\ref{fig:st-augment}:
\begin{itemize}
\item Add to the graph two copies of $P$, called $P_1$ and $P_2$,
  creating two new faces, one bounded by $P_1$ and $P$ and the other
  bounded by $P$ and $P_2$.  
\item Add $P_1$ and $P_2$ to $\opt$.
\end{itemize}
The resulting graph
is denoted $G_A$, and the resulting solution is denoted $\opt_A'$.
We consider $P_1$ and $P_2$ to be internal to the two bricks.  Thus
$P$ remains part of the boundary of each of the bricks, and $MG$
contains $P$ but not $P_1$ or $P_2$.  Since $P_1$ and $P_2$ share no
internal vertices with $P$, the joining vertices of $\opt_A'\cap B$ with
 $\partial B$ belong to $N_B$ and $S_B$.

\begin{figure}[ht]
  \centering
  \subfigure[]{\includegraphics[scale=.7]{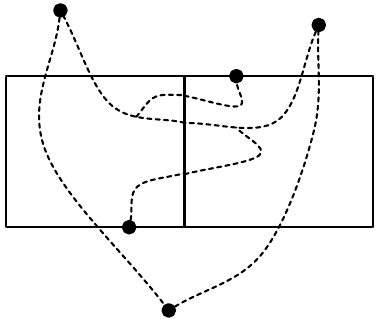}}
  \subfigure[]{\includegraphics[scale=.7]{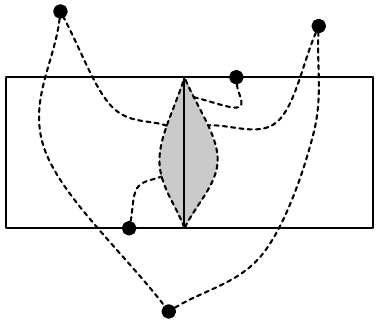}}
  \caption{Adding the column between two adjacent bricks (solid) in
    the {\em augment} step.  The dotted edges represent $\opt$ in (a)
    and $\opt_A$ in (b).}
  \label{fig:st-augment}
\end{figure}

We perform one more step, a {\em minimality-achieving step}:
\begin{itemize}
\item We remove edges from $\opt_A'$ until it is a minimal set of edges
  achieving the desired connectivity between terminals.
\end{itemize}
Let $\opt_A$ be the resulting set.  We get: 

\begin{lemma}\label{lem:joinNS}
 For every brick $B$, the joining vertices of $\opt_A\cap B$ with
 $\partial B$ belong to $N_B$ and $S_B$.
\end{lemma}

\subsubsection{Cleave} \label{sec:cleave}

We define
a graph operation, {\em cleave}.  Given a vertex $v$ and a bipartition
$A, B$ of the edges incident to $v$, $v$ is cleaved by
\begin{itemize}
\item splitting $v$ into two vertices, $v_A$ and $v_B$,
\item mapping the endpoint $v$ of edges in $A$ to $v_A$,
\item mapping the endpoint $v$ of edges in $B$ to $v_b$, and
\item introducing a zero-cost edge $e_v = v_Av_B$.
\end{itemize}
This operation is illustrated in Figure~\ref{fig:cleave}(a) and~(b).
If the bipartition $A,B$ is non-interleaving with respect to the
embedding's cycle of edges around $v$ then the construction maintains a planar embedding.

\setcounter{subfigure}{0}
\begin{figure}[ht]
  \centering
  \subfigure[]{\includegraphics{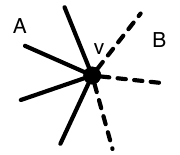}}
  \subfigure[]{\includegraphics{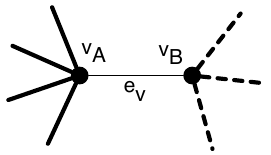}}\\
  \subfigure[]{\includegraphics{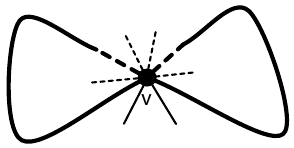}}
  \subfigure[]{\includegraphics{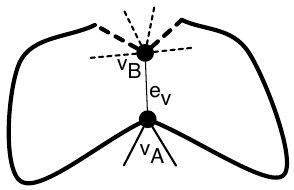}}
  \\
   \subfigure[]{\includegraphics{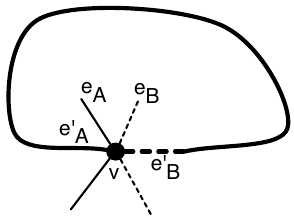}}
   \subfigure[]{\includegraphics{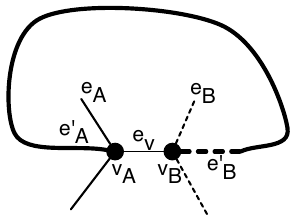}}
   \caption{Cleavings illustrated.  The bipartition of the edges
     incident to $v$ is given by the dashed edges $A$ and solid edges
     $B$. (a) Starting with this bipartition, (b) the result of
     cleaving vertex $v$ according to this bipartition. A simplifying
     cleaving of vertex $v$ with respect to a cycle (bold) before (c)
     and after (d).  A lengthening cleaving of a cycle (e) before and
     (f) after.  }
  \label{fig:cleave}
\end{figure}

We use two types of cleavings:
\begin{description}
\item[Simplifying cleavings] Refer to Figures~\ref{fig:cleave}(c)
  and~(d). Let $C$ be a clockwise non-self-crossing, non-simple cycle
  that visits vertex $v$ twice.  Define a 
  bipartition $A,B$ of the edges incident to $v$ as follows: given the
  clockwise embedding of the edges incident to $v$, let $A$ start and
  end with consecutive edges of $C$ and contain only two edges of
  $C$.  Such a bipartition exists because $C$ is non-self-crossing.
  
 \item[Lengthening cleavings] Refer to Figures~\ref{fig:cleave}(e)
   and~(f). Let $ C$ be a cycle, let $v$ be a vertex
   on $C$ with two edges $e_A$ and $e_B$ adjacent to $v$ embedded
   strictly inside $ C$, and let $e_A'$ and $e_B'$ be
   consecutive edges
   of $ C$ adjacent to $v$ such that the following bipartition is
   non-crossing with respect to the embedding: $A,B$ is a bipartition
   of the edges adjacent to $v$ such that $e_A, e_A' \in A$ and $e_B,
   e_B' \in B$.
\end{description}

We perform simplifying cleavings for non-simple cycles of $\opt_A$
until every cycle is simple; the artificial edges introduced are not included in
$\opt$.  The following lemma does not use planarity and shows that
(since cycles get mapped to cycles in this type of cleaving)
simplifying cleavings preserve two-edge connectivity.

\begin{lemma} \label{lem:OPT_C-2EC} Let $e$ be an edge in a graph $H$.
  Let $\widehat{H}$ be the graph obtained from $H$ by a simplifying
  cleaving.  Then $e$ is a cut-edge in $H$ iff it is a cut-edge in
  $\widehat{H}$.
\end{lemma}

\begin{proof} Let $u_1, u_2$ be the endpoints of $e$ and let $C$ be
  the cycle w.r.t.\ which a simplifying cleaving was performed.  If
  $H$ contains an $e$-avoiding $u_i$-to-$C$ path for $i=1,2$ then $e$
  is not a cut-edge in $H$, and similarly for $\widehat{H}$.  Suppose
  therefore that removing $e$ separates $u_i$ from $C$ in $H$.
  Then the same is true in $\widehat{H}$, and conversely.
\end{proof}

\begin{corollary} \label{cor:OPT_C-2EC} For $k=1,2$, if two vertices
  are $k$-edge connected in $H$ then any of their copies are
  $k$-edge connected in $H_C$.
\end{corollary}

Moreover, after all the simplifying cleavings, every cycle is simple, so:
\begin{lemma} \label{lem:OPT_C-2VC0}
  Vertices that are two-edge-connected in $\opt_C$ are biconnected.
\end{lemma}

Next we perform lengthening cleavings w.r.t.\ the boundary of a brick and
edges $e_A$ and $e_B$ of $\opt_C$; we include in $\opt_C$ all the
artificial zero-cost edges introduced.  
Lengthening cleavings clearly maintain connectivity.  Suppose that
vertices $x$ and $y$ are biconnected in $\opt_C$, and consider performing
a lengthening cleaving on a vertex $v$.  Since there are two
internally vertex-disjoint $x$-to-$y$ paths in $\opt_C$, $v$ cannot appear on
both of them.  It follows that there remain two internally
vertex-disjoint $x$-to-$y$ paths after the cleaving.  We obtain the
following lemma.

\begin{lemma} \label{lem:OPT_C-2VC}
  Lengthening cleavings maintain biconnectivity.
\end{lemma}

Lengthening cleavings are performed while there are still multiple
edges of the solution embedded in a brick that are incident to a
common boundary vertex.  Let $J_C$ be the set of artificial edges
that are introduced by simplifying and lengthening cleavings.
We denote the resulting graph by $G_C$, we denote the resulting mortar
graph by $MG_C$, and we denote the resulting solution by $\opt_C$.

As a result of the cleavings, we get the following:
\begin{lemma} \label{lem:OPT_C-tree-within-brick}
Let $B$ be a brick in $G_C$ with respect to $MG_C$.
The intersection $\opt_C \cap \interior(B)$ is a forest whose joining
vertices with $\partial B$ are the leaves of the forest.
\end{lemma}

\begin{proof}  Let $H$ be a connected component of $\opt_C \cap \interior(B)$.
  As a result of the lengthening cleavings, the joining vertices of
  $H$ with $\partial B$ have degree 1 in $H$.  Suppose otherwise; then there is
  a vertex $v$ of $H \cap \partial B$ that has degree $> 1$ in $H$.
  Hence $v$ is a candidate for a lengthening cleaving, a contradiction.

  By
  Theorem~\ref{thm:2vc-terminal} and the
  minimality-achieving step of the Augment step,  any cycle in $H$ must
  include a terminal $u$ with $r(u) = 2$ by 
  Theorem~\ref{thm:2vc-terminal}.  Since there are no terminals
  strictly enclosed by bricks, $u$ must be a vertex of $\partial B$.
  However, that would make $u$ a joining vertex of $H$ with $\partial
  B$.  As argued above, such vertices are leaves of $H$, a
  contradiction to the fact that $u$ is a vertex of a cycle in $H$.
  Therefore $H$ is acyclic.

  Furthermore, leaves of $H$ are vertices of $\partial B$ since $\opt_C$ is
  minimal with respect to edge inclusion and terminals are not
  strictly internal to bricks.
\end{proof}

\begin{lemma} \label{lem:cycles-preserved}
Let $C$ be a cycle in $\opt_C$.  Let
$B$ be a brick.  Distinct connected components of $C \cap \interior(B)$ belong
to distinct components of $\opt_C \cap \interior(B)$.
\end{lemma}

\begin{proof} Assume the lemma does not hold.  Then there is a
  $C$-to-$C$ path $P$ in $\interior(B)$.  Each vertex of
  $P$ that is strictly interior to $B$ is not a terminal.  A vertex of
  $P$ that was on $\partial(B)$ would be a candidate for a lengthening
  cleaving, a contradiction.  Therefore $P$ includes no terminals.
  This contradicts Theorem~\ref{thm:cycle-and-path}.
\end{proof}

\subsubsection{Flatten}

For each brick $B$, consider the edges of $\opt_C$ that are strictly
interior to $B$.  By Lemma~\ref{lem:OPT_C-tree-within-brick}, the
connected components are trees.  By Lemma~\ref{lem:joinNS}, for each
such tree $T$, every leaf is either on $B$'s northern boundary $N_B$
or on $B$'s southern boundary $S_B$.  For each such tree $T$ whose
leaves are purely in $N_B$, replace $T$ with the minimal subpath of
$N_B$ that contains all the leaves of $T$.  Similarly, for each such
tree $T$ whose leaves are purely in $S_B$, replace $T$ with the
minimal subpath of $S_B$ that contains all the leaves of $T$.

Let $\opt_F$ be the resulting solution.  Note that $\opt_F$ is a
multiset.  An edge of the mortar graph can appear with multiplicity
greater than one.

\subsubsection{Map}

This step is illustrated in Figures~\ref{fig:map}(a) and~(b).
In this step, the multiset $\opt_F$ of edges resulting from the flatten step
is used to select a set $\opt_M$ of edges of $\BC(MG_C)$.  Recall
that every edge $e$ of $MG_C$ corresponds in $\BC(MG_C)$ to three
edges: two brick copies (one in each of two bricks) and one
mortar-graph copy.  In this step, for every edge $e$ of $MG_C$, we
include the mortar-graph copy of $e$ in $\opt_M$ with multiplicity
equal to the multiplicity of $e$ in $\opt_F$.  At this point, none of
the brick copies are represented in $\opt_M$.

Next, recall that in the augment step, for each supercolumn $P$, we
created two new paths, $P_1$ and $P_2$, and added them to $\opt$.  The
edges of these two paths were not considered part of the mortar
graph, so mortar-graph copies were not included in $\opt_M$ for these
edges.  Instead, for each such edge $e$, we include the brick-copy of
$e$ in $\opt_F$ with multiplicity equal to the multiplicity of $e$ in
$\opt_F$.

Finally, for each edge $e$ interior to a brick, we include $e$ in
$\opt_M$ with the same multiplicity as it has in $\opt_F$.

\begin{figure}[ht]
  \centering
  \subfigure[]{\includegraphics{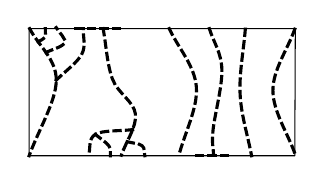}}
  \subfigure[]{\includegraphics{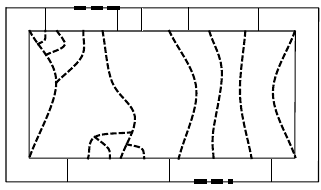}}
  \caption{Map (a) The intersection of $\opt_C$
    with a brick, dashed. (b) The same brick in the portal connected
    graph with portal edges (double-lines) connecting the brick to the
    corresponding face (outer boundary) of the mortar graph.}
  \label{fig:map}
\end{figure}

\subsubsection{Restructure} \label{sec:restructure}

Let $B$ be a brick.  For simplicity, we write the boundary paths of
$B$ as $N, E, S, W$. 
Let $F$ be the multiset of edges of $\opt_M$ that are in the interior
of $B$.  $F$ is a forest
(Lemma~\ref{lem:OPT_C-tree-within-brick}).  As a result of the {\em
  flatten} step, each component of $F$ connects $S$ to $N$.  We will replace $F$ with another subgraph $\widehat F$ and map each component $T$ of $F$ to a subgraph $\phi(T)$ of $\widehat F$ where $\phi(T)$ spans the leaves of $T$ and is a tree or a cycle.  Distinct components of $F$ are mapped by $\phi$ to edge-disjoint subgraphs (taking into account multiplicities).

Refer to Figure~\ref{fig:decomposition-paths}.  We inductively define  $S$-to-$N$ paths $P_0,P_1,\ldots $ and
corresponding integers $k_0, k_1, \ldots$.
Let $s_0, \ldots, s_t$ be the vertices of $S$ guaranteed
  by Lemma~\ref{lemma:brick-properties} (where $s_0$ is the vertex
  common to $S$ and $W$ and $s_t$ is the vertex common to $S$ and
  $E$).  Let $P_0$ be the easternmost path in $F$ from $S$ to $N$.
  Let $k_0$ be the integer such that $\pstart{P_0)}$ is in $S[s_{k_0}, s_{k_0+1})$.
 Inductively, for $i \geq 0$, let $P_{i+1}$ be the easternmost path in
$F_{N \wedge S}$ from $S[s_0,s_{k_i})$ to $N$ that is vertex-disjoint
from $P_i$.  Let $k_i$ be the integer such that $\pstart{P_i} \in S[s_{k_i}, s_{k_i+1})$.
This completes the inductive definition of $P_0,P_1, \ldots$.  Note
that the number of paths is at most $t$, which in turn is at most
$\kappa(\epsilon)$ as defined in Equation~\ref{eq:kappa}.

We use these paths to decompose $F$, as illustrated in
Figure~\ref{fig:decomposition-paths}.  
 Let $F_i$ be the set of edges of $F\setminus P_{i+1}$
  enclosed by the cycle formed by $P_i$, $P_{i+1}$, $N$ and $S$.
  Clearly $F = \cup_i F_i$.  If $P_i$ is connected to $P_{i+1}$, they
  share at most one vertex, $w_i$.  If they are not connected, we say
  $w_i$ is undefined.

There are two cases: either $F_i$ is connected or not.

\begin{figure}[ht]
  \centering
  \includegraphics[scale=2.5]{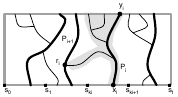}
  \caption{Paths used to decompose $F$.  The brick
    boundary is given by the rectangle.  The paths $P_0, P_1, \ldots$
    are bold.  $F_i$ is given by the shaded background.}
  \label{fig:decomposition-paths}
\end{figure}

\paragraph{Connected case:}  There are two subcases.
 Either $F_i$ spans vertices of $S[\cdot, s_{k_i})$ or not.

\bigskip

Suppose $F_i$ spans vertices of $S[\cdot, s_{k_i})$.  
Let $T_S$ be a
minimal subtree of $F_i$ that spans $F_i \cap S$ and let $T_N$ be a
minimal subtree of $F_i$ that spans $F_i \cap N$.  Let $r_N$ be the
first vertex of $P_i$ in $T_N$ and let $r_S$ be the last vertex of
$P_i$ in $T_S$.  A path in $F_i$ from $T_S$ to $N$ that does not go
through $r_S$ contradicts the choice of $P_{i+1}$ as there would be a
path in $F_i$ from $S[\cdot, s_{k_i})$ to $N$ that is disjoint from
$P_i$.  It follows that $T_S$ and $T_N$ are edge disjoint: if they
intersect they may only do so at $r_S = r_N$.  If $w_i$ is defined,
then there is a path $Q$ from $w_i$ to $P_i$; $Q$ intersects $P_i$
between $r_N$ and $r_S$, for otherwise there would be a superpath of
$Q$ that contradicts the choice of $P_{i+1}$.

If $w_{i-1}$ is defined and $w_{i-1} \in T_N$, then we replace $T_N$
with the tree guaranteed by Lemma~\ref{lem:short-runs} with roots
$r_N$ and $w_{i-1}$.  Otherwise, we replace $T_N$ with the tree
guaranteed by Lemma~\ref{lem:fib-tree} with root $r_N$. We do the
same for $T_S$. 

\bigskip

Suppose $F_i$ does not span vertices of $S[\cdot, s_{k_i})$.   
Let
$T_N$ be a minimal connected subgraph of $F_i \cup S[s_{k_i},\pstart{P_i}]$ that spans
$F_i \cap N$.  Let $r_N$ be the first vertex of $P_i$ in $T_N$.  If
$w_i$ is defined, then there is a path $Q$ from $w_i$ to $P_i \cup
S[s_{k_i},\pstart{P_i}]$ and $Q$'s intersection with $P_i$ belongs to
$P_i[\cdot, r_N]$, for otherwise
there would be a superpath of $Q$ that contradicts the choice of
$P_{i+1}$.  If $w_{i-1}$ is defined and $w_{i-1} \in T_N$, then we
replace $F_i$ with the tree guaranteed by Lemma~\ref{lem:short-runs}
with roots $r_N$ and $w_{i-1}$ along with $Q$, $P_i[\cdot, r_N]$, and
$S[s_{k_i},\pstart{P_i}]$.  Otherwise we replace $F_i$ with the tree guaranteed
by Lemma~\ref{lem:fib-tree} with root $r_N$ along with $Q$,
$P_i[\cdot, r_N]$, and $S[s_{k_i},\pstart{P_i}]$.

\bigskip

In both cases, we define $\phi'(F_i)$ to be the resulting tree that replaces $F_i$.  By construction, $\phi'(F_i)$ spans the leaves of $F_i$ and $w_{i-1}$ and $w_i$ (if defined).

\paragraph{Disconnected case:} In this case, by the definition of
$P_{i+1}$, $F_i \cap S$ is a subset
of the vertices of $S[s_{k_i},\pstart{P_i}]$, for otherwise there would be a
path to the right of $P_{i+1}$ that connects to $N$ and is disjoint
from $P_i$.

If $F_i$ is connected to $F_{i+1}$, then the western-most tree $T_W$ is a
tree with root $w_i$ and leaves on $S$ and does not connect to $N$ as
that would contradict the choice of $P_{i+1}$; if this is the case, let
$\widehat T_W$ be the tree guaranteed by Lemma~\ref{lem:fib-tree} and define $\phi'(T_W) = \widehat T_W$. 

If $F_i$ is connected to $F_{i-1}$, let $S'$ be the subpath of $S$ that
spans the eastern-most tree $T_E$'s leaves on $S$.  Let $\widehat T_E$ be
the tree guaranteed by Lemma~\ref{lem:short-runs} that spans the
eastern-most tree's leaves on $N$ and roots $w_{i-1}$ and $\pstart{P_i}$ and define $\phi'(T_E) = \widehat T_E$.

Let $\mathcal F$ be the set of remaining trees, let $P =
N$, and let $Q = S[s_{k_i},\pstart{P_i}]$ in
Lemma~\ref{lem:forest-cycle}.  Let $\widehat C$, $\widehat{\mathcal
  F}$, and $\phi$ be the cycle (or empty set), set of trees, and
mapping that satisfy the properties stated in the lemma.

We define $\widehat F_i$ to consist of the trees of $\widehat{\mathcal F}$ and the
cycle $\widehat C$ (and $\widehat T_W$, $S'$ and $\widehat
T_E$ if defined).

\bigskip 

\noindent We replace every $F_i$ with $\widehat F_i$, as described
above, for every brick, creating $\opt_S$.  This is illustrated in
Figure~\ref{fig:restructure}(a).  Now we define $\phi$ in terms of
$\phi'$.  A component $T$ of $F$ is partitioned into adjacent trees
in this restructuring: namely $T_1, \ldots, T_k$, $k \ge 1$. 
$T_1$ and $T_k$ may be restructured via the disconnected case and all others are restructured via the connected case.  Define $\phi(T) = \cup_{i=1}^k \phi'(T_i)$.  If $k > 1$, then consecutive trees $T_i$ and $T_{i+1}$ share a vertex $w_i$ and by construction $\phi'(T_i)$ and $\phi'(T_i)$ also share this vertex.  Since $\phi'(T)$ spans the leaves of $T$, we get that $\phi(T)$ spans the leaves of $T$, as desired.  Also by construction, the submapping of $\phi'$ of trees to trees (and not cycles) is bijective; the same holds for $\phi$.

\paragraph{Number of joining vertices} In both the connected and
disconnected case, the number of leaves is the result of a constant
number of trees resulting from
Lemmas~\ref{lem:fib-tree},~\ref{lem:short-runs}
and~\ref{lem:forest-cycle}.  Therefore, $\widehat F_i$ has
$o(\epsilon^{-2.5})$ joining vertices with $N$ and $S$.  Since $i \le
\kappa(\epsilon) = O(\epsilon^{-3})$, $\opt_S$ has
$o(\epsilon^{-5.5})$ joining vertices with the boundary of each
brick.  This is the number of connections required to allow a solution to be nearly optimal and affects the number of portals required in Equation~(\ref{eq:theta}):
\begin{equation}
  \label{eq:alpha}
  \alpha(\epsilon) \mbox{ is } o(\epsilon^{-5.5})
\end{equation}

This will allow us to prove the second part of the Structure Theorem.

\begin{figure}[ht]
  \centering
  \subfigure[]{\includegraphics{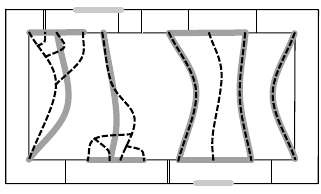}}
  \subfigure[]{\includegraphics{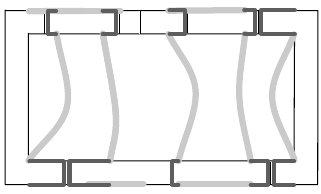}}
  \caption{Continuing from Figure~\ref{fig:map}, restructure and
    rejoin: (a) Restructured version (dark grey) of the intersection
    of $\opt_M$ with the brick (dashed).  (b) Connecting the
    restructured solution inside the brick to the mortar graph through
    portals (via dark grey edges).}
  \label{fig:restructure}
\end{figure}

\subsubsection{Rejoin}

In this step, we make inter-brick connections for parts that were
disconnected in the {\em mapping} step.  Since joining vertices
represent the ends of all disconnected parts, it suffices to connect
joining vertices of $\opt_S$ with $\partial B$ to their mortar-graph
counterparts via portal edges.

This is illustrated in Figure~\ref{fig:restructure}~(d): We first move
the edges of $\opt_S \cap \partial B$ to $MG$ for every brick $B$.
Such edges may have been introduced in the restructure step: for
every brick $B$, we connect $\opt_S \cap B$ to the mortar graph.  For
every joining vertex $v$ of $\opt_S \cap B$, we find the nearest portal
$p_v$, add the subpath of $\partial B$ and $MG$ connecting $v$ and
$p_v$ and add the portal edge corresponding to $p_v$.  (We need at
most two copies of each portal edge.) Finally we
contract the edges introduced in the cleaving step.  This
produces a solution $\widehat{OPT}$ of $\mathcal{B}^+(MG)$.

\subsection{Analysis of connectivity} \label{sec:connectivity-analysis}

In the augment step, because the added paths $P_1$ and $P_2$ form a
cycle, this transformation preserves two-connectivity and connectivity
between terminals.  Cleaving clearly preserves connectivity and, by
Lemmas~\ref{lem:OPT_C-2VC0} and~\ref{lem:OPT_C-2VC}, terminals that
require two-connectivity are biconnected in $\opt_C$.  Therefore, for
terminals $x$ and $y$ requiring connectivity, there is a path $P_C$ in
$\opt_C$ connecting them.  If $x$ and $y$ require two-connectivity,
there is a simple cycle $C_C$ in $\opt_C$ connecting them.  We follow
$P_C$ and $C_C$ through the remaining steps.

\begin{description}
\item[Flatten] Consider a tree $T$ that is replaced by a subpath $Q$
  of a northern or southern boundary of a brick that spans $T$'s
  leaves.  $Q$ spans any terminals that $T$ spanned, since there are
  no terminals internal to bricks.  

  $P_C \cap T$ is therefore a (set of) leaf-to-leaf paths and so $(P_C
  - T) \cup Q$ contains an $x$-to-$y$ path.  It follows that there is
  an $x$-to-$y$ path $P_F$ in $\opt_F$.

  By Lemma~\ref{lem:cycles-preserved}, $C_C \cap T$ is a single path
  and, by the above reasoning, is a leaf-to-leaf path.  Therefore
  $(C_C -T ) \cup Q$ contains a cycle through $x$ and $y$.  It follows
  that there is a cycle $C_F$ through $x$ and $y$ in $\opt_F$

\item[Map] $P_F$ ($C_F$) gets mapped to a sequence ${\cal P}_M =
  (P_M^1, P_M^2, \ldots )$, (a cyclic sequence ${\cal C}_M = (C_M^1,
  C_M^2, \ldots )$) of paths of $\opt_M$ in $\BC(MG_C)$ such that each
  path either consists completely of mortar-graph edges or consists
  completely of brick-copy edges.  The last vertex of one path and the
  first vertex of the next path are copies of the same vertex of
  $G_C$, and that vertex belongs to a north or south boundary of
  $MG_C$.  By Lemma~\ref{lem:OPT_C-tree-within-brick}, each path in
  ${\cal P}_M$ or ${\cal C}_M$ that consists of brick-copy edges
  starts and ends at the northern or southern boundary of a brick.
  
\item[Restructure] We define a mapping $\hat \phi$, based in part on the map
  $\phi$ defined in
  Section~\ref{sec:restructure}.
 For a path $Q$ in ${\cal P}_M$ or ${\cal C}_M$ that
  uses mortar-copy edges, define $\hat \phi(Q) = Q$.  For a path $Q$
  in ${\cal P}_M$ or ${\cal C}_M$ that uses brick-copy edges, let $T$
  be the tree in $\opt_M$ that contains $Q$ and define $\hat \phi(Q) =
  \phi(T)$.

  Let ${\cal C}_S$ be the cyclic sequence of trees and cycles to which
  ${\cal C}_M$ maps by $\phi$.  Since $\phi(T)$ spans the leaves
  (joining vertices) of $T$, consecutive trees/cycles in ${\cal C}_S$
  contain copies of the same vertex.  By Lemma~\ref{lem:cycles-preserved} and
  that, within a brick, the preimage of the set of trees mapped to by
  $\phi$, the trees of ${\cal C}_S$ are edge disjoint.  (The 
  cycles may be repeated.)

  Likewise we define ${\cal P}_S$ having the same properties except for the fact that the sequence is not cyclic.

\item[Rejoin] This step reconnects ${\cal P}_S$ and ${\cal C}_S$.

  Consider a tree $T$ in either of these sequences that contains
  brick-copy edges. The rejoin step first moves any edge of $T$ that
  is in a brick boundary to the mortar copy. It then connects 
  joining vertices to the mortar by way of detours to portals and
  portal edges. Therefore, $T$ is mapped to a tree $T_J$ that connects
  the mortar copies of $T$'s leaves.

  Consider a cycle $C$ in either ${\cal P}_S$ or ${\cal C}_S$.  $C$
  contains a subpath of $N$ and $S$ whose edges are moved to their
  mortar copy and two $N$-to-$S$ paths whose joining vertices are
  connected to their mortar copies.  Therefore $C$ is mapped to a
  cycle $C_J$ through the mortar copies of the boundary vertices of $C$.

  Let the resulting sequence and cyclic sequence of trees and cycles
  be ${\cal P}_J$ and ${\cal C}_J$.  Since consecutive trees/cycles in ${\cal C}_S$
  contain copies of the same vertex, in ${\cal C}_J$ consecutive trees/cycles
contain 
  common mortar vertices.  We have that ${\cal C}_J$
  is a cyclic sequence of trees and cycles, through $x$ and $y$, the
  trees of which are edge-disjoint.  Therefore the union of these
  contains a cycle through $x$ and $y$.

  Similarly, we argue that the union of the trees and cycles in ${\cal
    P}_J$ contains and $x$-to-$y$ path.

\end{description}

\subsection{Analysis of cost increase}

By Lemma~\ref{lem:col-length}, the total costs of all the east and
west boundaries of the bricks is an $\epsilon$ fraction of $\opt$, so
we have
\begin{equation}
  \label{eq:OPT1}
  \cost(\opt_A) \leq (1+2\epsilon)\cost(\opt).
\end{equation}
The cleaving step only introduces edges of zero cost, so
\begin{equation}
  \label{eq:OPT2}
  \cost(\opt_C) = \cost(\opt_A).
\end{equation}
The flatten step replaces trees by $\epsilon$-short paths, and so can
only increase the cost by an $\epsilon$ fraction, giving:
\begin{equation}
  \label{eq:OPT3}
  \cost(\opt_F) \le (1+\epsilon) \cost(\opt_C).
\end{equation}
The mapping step does not introduce any new edges, so
\begin{equation}
  \label{eq:OPT4}
  \cost(\opt_M) = \cost(\opt_F).
\end{equation}
The restructure step involves replacing disjoint parts of $\opt_M \cap
B$ for each brick by applying
Lemmas~\ref{lem:fib-tree},~\ref{lem:short-runs},
and~\ref{lem:forest-cycle}.  This increases the cost of the solution
by at most an $O(\epsilon)$ fraction.  Further we add subpaths of
$S[s_{k_i},\pstart{P_i}]$ where $\opt_M$ contains disjoint subpaths $P_i$ from
$\pstart{P_i}$ to $N$ and by the brick properties, $\cost(S[s_{k_i},\pstart{P_i}]) \le
\cost(P_i)$.  This increases the the cost of the solution
by at most another $O(\epsilon)$ fraction.  We get, for some constant $c$:
\begin{equation}
  \label{eq:OPT5}
  \cost(\opt_S) \leq (1+c\epsilon)\cost(\opt_M).
\end{equation}
In the rejoin step, we add two paths connecting a joining vertex to
its nearest portal: along the mortar graph and along the boundary of
the brick.  The cost added per joining vertex is at most twice the
interportal cost: at most $2\cost(\partial B)/\theta$ for a joining
vertex with the boundary of brick $B$, by
Lemma~\ref{lem:portal-distance}.  Since each mortar graph edge appears
as a brick-boundary edge at most twice, the sum of the costs of the
boundaries of the bricks is at most $18\epsilon^{-1}\opt$
(Equation~(\ref{eq:MG-length})). Since there are $\alpha(\epsilon)$
joining vertices of $\opt_S$ with each brick, the total cost added
due to portal connections is at most $36\frac{\alpha}{\theta\epsilon} \opt$.  Replacing $\alpha$ and $\theta$ via Equations~(\ref{eq:alpha}) and~(\ref{eq:theta}) gives us that the total cost added due to portal connections is
\begin{equation}
  \label{eq:OPT6}
  O(\epsilon \cost(\opt))
\end{equation}
Combining equations~(\ref{eq:OPT1}) through~(\ref{eq:OPT6}),
$\cost(\widehat\opt) \leq (1+c'\epsilon)\cost(\opt)$ for some constant $c'$,
proving Theorem~\ref{thm:structure}.

\section{Dynamic Program} \label{sec:dp}

In this section we show that there is a dynamic program that finds an
approximately optimal solution in the portal-connected graph $\BC(P)$.

\subsection{Structure of solution} \label{sec:structure-of-solution}

We start by showing that we can restrict our attention to solutions
whose intersection with each brick is a small collection of trees
whose leaves are portals.

\begin{lemma} \label{lem:number-of-trees} Let $S$ be a minimal
  solution to the instance $(\BC(MG), \rvec)$ of the $\set{0,1,2}$-edge
  connectivity problem.  Let $B$ be a brick, and suppose that the
  intersection of $S$ with $B$ is the union of a set of non-crossing
  trees whose leaves are portals.  Then the number of such trees is at
  most three times the number of portals.
\end{lemma}

\begin{proof} Let the portals be $v_1, \ldots, v_k$ in
  counterclockwise order.  Each tree induces a subpartition on the set of
  pairs $\set{(i, i+1)\ : i \in \set{1,\ldots, k-1}}$
  as follows: if the leaves of the tree are $v_{i_1}, v_{i_2}, \ldots,
  v_{i_p}$ where $i_1 < i_2 < \cdots < i_p$, then the parts are
$$\set{(i_1,i_1+1), \ldots, (i_2-1, i_2)}, \set{(i_2, i_2+1), \ldots,  (i_3-1,i_3)}, \ldots, \set{(i_{p-1},i_{p-1}+1), \ldots, (i_p-1, i_p)}$$
Because the trees are non-crossing, the corresponding subpartitions
are non-crossing as well: for each pair $T_1, T_2$ of trees, either
each part of $T_2$ is a subset of some part of $T_1$, or vice
versa.

Therefore the trees themselves form a rooted tree $\cal T$ according
to the nesting relation of the corresponding subpartitions.  Furthermore, by
minimality of $S$, no three trees have the same sets of leaves.  This
shows that a node of $\cal T$ with only one child has a parent with at
least two children.  The
number of leaves of $\cal T$ is at most the number of pairs $(i,i+1)$,
which is at most $k-1$.  This shows that the number of trees is at most $3(k-1)$.
\end{proof}

\begin{corollary} \label{lem:crossing-trees} There exists a function
  $f$ that maps each brick $B$ to a cardinality-at-most-$3\theta$ subpartition
  $f(B)$ of the portals of $B$ with the following property:\\
  for each brick $B$, for each part $P$ of $f(B)$, let $T_P$ be any
  minimum-cost tree in $B$ whose leaves are $P$, and let $H_B$ be the
  union $\bigcup_{P\in f(B)} T_P$.  Then there is a feasible solution
  $S'$ to the instance $(\BC(MG), \rvec)$ such that
\begin{itemize}
  \item the cost of $S'$ is at most $(1+c\epsilon)\opt$ where $c$ is
    an absolute 
    constant, and
  \item the intersection of $S$ with any brick $B$ is $H_B$.
  \end{itemize}
\end{corollary}

\begin{proof} Theorem~\ref{thm:structure} states that there is a
  feasible solution $S$ of cost at most $(1+c\epsilon)\opt$ such that
  the intersection of $S$ with any brick $B$ is the multi-set union of a family
  ${\cal T}_B$ of
  non-crossing trees whose leaves are portals. We assume that
  $S$ is a minimal feasible solution.  
  Lemma~\ref{lem:number-of-trees} shows that $|{\cal T}_B| \leq
  3\theta$.

  Now we construct the multi-subgraph $S'$ from $S$.  For each brick
  $B$, replace each tree in ${\cal T}'_B$ with a minimum-cost tree having the
  same leaves.  Clearly the cost of $S'$ is at most that of $S$.  It
  remains to show that $S'$ is a feasible solution.

  Let $u$ and $v$ be two terminals.  Let $P$ or $C$ be a minimal
  $u$-to-$v$ path or minimal cycle containing $u$ and $v$ in $\BC(MG)$
using edges of $S$.  We obtain a $u$-to-$v$ path $P'$
  in $S'$ or a cycle $C'$ in $S'$ containing $u$ and $v$ as follows.
  For each brick $B$, the intersection of $P$ or $C$ with $B$ is a set
  of paths $P_1, \ldots, P_k$ through the trees in ${\cal T}_B$ such
  that each path $P_i$ starts and ends on a portal.  For each path
  $P_i$, the tree in ${\cal T}_B$ connecting the endpoints of $P_i$ is
  replaced in $S'$ by another tree that includes the same endpoints,
  so that tree contains a path $P_i'$ with the same endpoints.  We
  replace each path $P_i$ with the path $P_i'$.  Let $P'$ and $C'$ be
  the subgraphs obtained by performing these transformations for each
  brick $B$.  The transformations ensure that $P'$ and $C'$ are in
  $S'$.  The path $P'$ shows that $u$ and $v$ are connected in $S'$.

  For a cycle $C$ in $S$, we still need to show that $C'$ is a cycle.
  In particular, we need to show that, for each brick $B$, the paths
  $P_1, \ldots, P_k$ forming the intersection of $C$ with $B$ all
  belong to different trees in ${\cal T}'_B$.  Assume for a
  contradiction that $P_i$ and $P_j$ belong to a single tree $T\in
  {\cal T}'_B$.  By the assumption that the degree is at most three,
 $P_i$ and $P_j$ cannot share a
  vertex.  Therefore there is a $P_i$-to-$P_j$ path in $T$ containing
  at least one edge.  However, since $C$ is a cycle containing
  $P_i$ and $P_J$, an edge of the $P_i$-to-$P_j$ path can be removed
  while preserving feasibility, a contradiction.
\end{proof}

\subsection{The tree used to guide the dynamic program} \label{sec:tree}

 Recall that $\CC(P)$ is the brick-contracted graph, in which each brick of $\BC(P)$ is contracted to a vertex.
Recall that a parcel $P$ is a subgraph of $MG$ and defines a set of
bricks contained by the faces of $P$.  The planar dual of the parcel
has a spanning tree $T$ of depth $\eta+1$.  The algorithm transforms $T$ into a
spanning tree $T'$ of the planar dual of $\CC(P)$ as follows.  Each
brick is a vertex of $T$; replace that vertex with a cycle minus one
edge consisting of the duals of the portal edges, as shown in
Figure~\ref{fig:from-tree-to-tree}.  

\begin{figure}
\centerline{\includegraphics[scale=0.6]{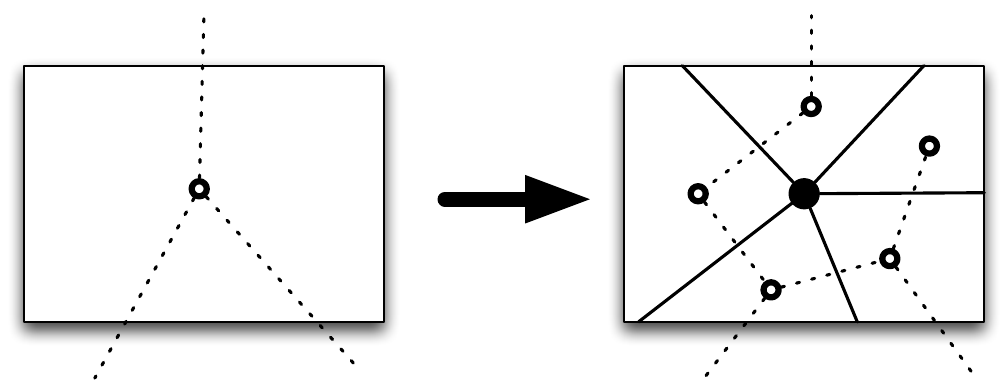}}
\caption{The spanning tree of the dual of the parcel is transformed to
be a spanning tree of the planar dual of $\CC(P)$.  For each brick,
all but one of the portal edges is included in the new spanning tree.}
\label{fig:from-tree-to-tree}
\end{figure}

Since each brick has at most
$\theta$ portal edges, it follows that the spanning tree $T'$ of the dual of $\CC(P)$ has 
depth at most $(\eta+1)\theta$.   Next, the algorithm defines $\widehat
T$ to be the set of edges of $\CC(P)$ whose duals are not in $T'$.  A
classical result on planar graphs implies that $\widehat T$ is a
spanning tree of $\CC(P)$.  The construction ensures that each vertex
of $\CC(P)$ that 
corresponds to a brick has only one incident edge in $\widehat T$.
By our assumption that each vertex in the input graph has degree
at most three, every vertex of $\CC(P)$ that appears in $G$ has degree
at most three in $\widehat T$.   Thus $\widehat T$ has degree at most three.
The bound on the depth of $T'$ ensures that, 
for each vertex $v$ of $\widehat T$, the graph $\CC(P)$ contains at most $2(\eta+1)(\theta+1)+1$ edges between
descendents of $v$ and non-descendents.

\subsection{The dynamic programming table} \label{sec:dp-table}

The algorithm considers the spanning tree $\widehat T$ of $\CC(P)$ as
rooted at an arbitrary leaf.  By our assumption that the input graph
has degree three, each vertex $v$ of $\widehat T$ has at most two
children.  Let $\widehat T(v)$ denote the
subtree of $\widehat T$ rooted at $v$.  For each vertex $v$ of $\widehat
T$, define
$$f(v) =   \left\{
\begin{array}{ll}
 B & \mbox{if $v$ is the result of contracting a brick } B\\
 v & \mbox{otherwise}
\end{array}
\right.$$ and define $W(v)$ to be the subgraph of $\BC(P)$ induced by
$\bigcup \set{f(w) : w \in \widehat T(v)}$.  Let $\delta(S)$ be the
subset of edges with exactly one endpoint in the subgraph $S$ (i.e.~a
cut).  It follows that the cut $\delta(W(v))$ in $\BC(P)$ is equal to
the cut $\delta(\widehat T(v))$ in $\CC(P)$ and so has $O(\theta\eta)$
edges.  The dynamic programming table will be indexed over
configurations, defined below.

\subsubsection{Configurations} 

Let $L = \delta(H)$ for a connected, vertex-induced subgraph $H$ of $\BC(P)$.  We define a {\em
  configuration} $K_L$ corresponding to $L$, illustrated in
Figure~\ref{fig:2-ec-configs}.  First, for each edge $e\in L$, we
defined two
  edges $e^1$ and $e^2$ to reflect the fact that an edge can be used twice.
Let $\hat L$ be the resulting set of edges, $\hat L=\bigcup_{e\in L}\set{e^1,e^2}$.
 A {\em configuration}
is a forest with no degree-2 vertices whose leaf edges are a
subset of $\hat L$, and such that, for each edge $e\in L$, if the forest
contains both $e^1$ and $e^2$ then these two edges are incident to the
same vertex in the forest, together with a $\set{1,2}$ labeling of the internal
vertices.
We denote the set of all configurations on edge set $L$ by
${\cal K}_L$.

\begin{figure}[ht]
  \centering
  \subfigure[]{\includegraphics{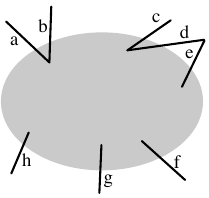}}\qquad
  \subfigure[]{\includegraphics{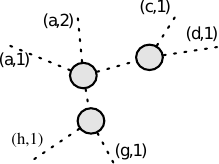}}\qquad
  \subfigure[]{\includegraphics{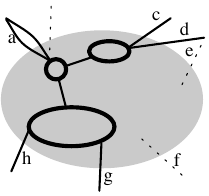}}
  \caption{(a) The cut edges  (black) of a subgraph $H$ (shaded).
    (b) A configuration (dotted forest) for $L=\delta(H)$.  (c) A subgraph (bold) that meets the
    configuration.}
  \label{fig:2-ec-configs}
\end{figure}

\begin{lemma}\label{lem:dp-size}
  The number of configurations for $H$ where $n=|\delta(H)|$  is at most $16^n(2n)^{2n-2}$
  and these trees can be computed in $O(16^n(2n)^{2n-2})$ time.

\end{lemma}

\begin{proof}
 A configuration can be selected as follows.  First, for each
  of the $n$ edges, decide whether edge $e^1$ is to be included and whether
  $e^2$ is to be included.  (There are $4^n$ choices.) Let $n'$ be the
  number of edges $e$ for which either $e^1$ or $e^2$ is to be
  included.  Next, select a tree with the selected edges as leaf
  edges.  It follows from Cayley's formula that the number of such
  trees is $2n'^{2n'-2}$, which is at most $2n^{2n-2}$.    Next,
  delete some subset of non-leaf edges of the tree.  There are at most
  $2^n$ ways of selecting such a subset.  Finally, select a
  $\set{1,2}$ labeling of the internal vertices.  There are at most
  $2^n$ such labelings.

  The set of trees can be computed in $O(1)$ amortized time per
  tree~\cite{NU03}.
\end{proof}

\paragraph{Connecting} A configuration $K_L \in {\cal K}_L$ is {\em
  connecting} if, for each internal vertex $v$ of $K_L$, if $v$ has
label $c$ then $K_L$ contains $c$ paths from $v$ to leaves.

\paragraph{Compatibility} Configurations $K_A \in {\cal K}_A$ and $K_B
\in {\cal K}_B$ are {\em compatible} if for every edge $e \in \hat A \cap
\hat B$ either $e \in K_A \cap K_B$ or $e \notin K_A \cup K_B$.

\paragraph{Compressed bridge-block forest}

For a graph $G$, a {\em block} is a maximal subgraph such that, for
every two vertices $u$ and $v$ in the subgraph, there are two
edge-disjoint $u$-to-$v$ paths.  Contracting each block of $G$ to a
single vertex yields the {\em bridge-block} forest of $G$.  It is easy
to see that it is indeed a forest.  An edge $e$ of $G$ is an edge of this
forest if $e$ is a bridge of $G$, i.e. if the endpoints of $e$ are in
different connected components of $G-e$.

We define the {\em compressed} bridge-block forest to be the forest
obtained from the bridge-block forest by substituting an edge for each maximal path
of internal degree two.  We denote the compressed bridge-block forest
of $G$ by $\widetilde{EC}(G)$.

\paragraph{Consistency}

We say a configuration $K_A$ is {\em consistent} with a set of
mutually compatible configurations $\set{K_{A_1}, K_{A_2}, \ldots}$ if
\begin{itemize}
\item there is an isomorphism between $K_A$ and $\widetilde{EC}(\cup_i
  K_{A_i})$ that preserves the identity of leaf edges of $K_A$, and 
\item for each vertex $x$ of $\cup_i K_{A_i}$ that is labeled~2, $x$ is
in a block of $\cup_i K_{A_i}$ that corresponds to a vertex in $K_A$
that is labeled~2.
\end{itemize}

\paragraph{Meeting} Let $H$ be a connected, vertex-induced subgraph of
$G$.  Let $M$ be a minimal solution to an instance $(G, \rvec)$.  Let $M_H$ be the graph obtained from $M$ as 
follows.  Remove edges not in $H \cup
\delta(H)$.  Next, for each vertex $v$ outside of $H$, if $v$ has $k$
incident edges, replace $v$
with $k$ copies of $v$, one incident to each of these edges. 

We say $M$ {\em meets} a configuration
$K_{\delta(H)}$ if $\widetilde{EC}(M_H) = K_{\delta(H)}$ and if, for each terminal $x$ in $H$,
$M_H$ either contains $r(x)$ edge-disjoint paths to vertices
outside of $H$ or contains $\min\set{r(x), r(y)}$ edge-disjoint paths
to every other terminal $y$.

\paragraph{DP table entry}
The dynamic program constructs a table $DP_v$ for each vertex $v$ of
$\widehat T$.  The table $DP_v$ is indexed by the configurations of
$\delta(W(v))$.  We showed in Section~\ref{sec:structure-of-solution}
that we can restrict our attention to solutions with the following
property: 
\begin{quote}
the intersection
with each brick is a cardinality-at-most-$3\theta$ collection of minimum-cost trees
whose leaves are portals. 
\end{quote}
For each configuration $K$ of
$\delta(W(v))$, the entry $DP_v[K]$ is the minimum cost of a subgraph
of $W(v)$ that meets configuration $K$ and has the above property.
We do not count the cost of edges of $\delta(W(v))$.

\subsubsection{The filling procedure}

If $u$ is not a leaf of $\widehat T$, then we populate the entries of
$DP_u$ with the procedure {\sc fill}.  We use the shorthand
$K$ for $K_{\delta(\set{u})}$ and $K_i$ for $K_{\delta(W(u_i))}$.
The cuts are with respect to the graph $\BC(P)$.

\begin{tabbing}
  {\sc fill}$(DP_u)$ \\
  \qquad \= Initialize each entry of $DP_u$ to $\infty$. \\
  \> Let $u_1, \ldots, u_s$ be the children of $u$. \\
  \> For every set of connecting, mutually compatible
  configurations $K, K_1, \ldots, K_s$,\\
  \> \qquad \= For every connecting configuration $K_0$ that is
  consistent with $K, K_1, \ldots, K_s$, \\
  \> \> \qquad \= $\mbox{cost} \leftarrow \cvec(K \cap(\cup_{i=1}^s
  K_i))+(\cvec(K_1 \cap K_2) \text{ if } s=2 \text{ else } 0)
+ \sum_{i=1}^s DP_{u_i}[K_i]$. \\
  \> \> \> $DP_u[K_0] \leftarrow \min \{DP_u[K_0], \mbox{cost}\}$
\end{tabbing}

Since $\widehat T$ has degree at most three, the number $s$ of
children of each vertex $u$ in $\widehat T$ is at most two.

If $u$ is a leaf of $\widehat T$ and $u$ does not correspond to a
brick (i.e.~$f(u) = u$), the problem is trivial.  Each configuration
$K$ is a star: $u$ is the center vertex, and the edges of
$\delta(\set{u})$ are the edges of $K$.  Since the cost of any
subgraph that is induced by $\set{u}$ is zero, the value of $DP_u[K]$ is zero for every $K$.  

Suppose $u$ is a leaf of $\widehat T$ and $f(u)$ is a brick $B$.
Recall that we restrict our attention to solutions whose intersection with
$B$ is a collection of at most $3\theta$ minimum-cost trees whose leaves are portals.
The algorithm populates the table $DP_u$ as follows.  Initialize each
entry of $DP_u$ to $\infty$.  Next, iterate over cardinality-at-most-$3\theta$
families of subsets of the portals.  For each family $\cal F$,
\begin{itemize}
\item  define a
subgraph $H_{\cal F}$ to be the multi-set union over subsets $P$ in $\cal F$ of
a minimum-cost tree spanning $P$, 
\item find the configuration $K$ corresponding to $H_{\cal F}$, and
\item set $DP_u[K] \leftarrow \min \set{DP_u[K], \text{cost of }
    H_{\cal F}}$.
\end{itemize}
The minimum-cost tree spanning a subset of portals can be computed in
$O(\theta^3 n)$ time using
the algorithm of Theorem~\ref{thm:emv87}.

\subsubsection{Running time}

Consider the time required to populate the $DP_u$ for all the leaves
$u$ of $\widehat T$.  We need only consider non-crossing partitions of
subsets of $\delta(W(v))$ since $H_K$ is the union of non-crossing
trees.  The number of non-crossing partitions of an $n$ element
ordered set is the $n^{th}$ Catalan number, which is at most
$4^n/(n+1)$.  Therefore, the number of non-crossing sub-partitions is
at most $4^n$.  It follows that the time to populate $DP_v$ for $v$ a
brick vertex is $O(\theta^4 4^\theta |B|)$ which is
$O(2^{\mathrm{poly}(1/\epsilon)}|B|)$ since $\theta$ depends
polynomially on $1/\epsilon$.  Since a vertex appears at most twice in
the set of bricks, the time needed to solve all the base cases in
$O(2^{\mathrm{poly}(1/\epsilon)}n)$ where $n$ is the number of
vertices in the parcel.

Consider the time required to populate the $DP_u$ for all the internal
vertices $u$ of $\widehat T$. The number of edges in $\delta(W(v))$ in
$\BC(P)$ is $O(\theta\eta)$.  By Lemma~\ref{lem:dp-size}, it follows
that the corresponding number of configurations is
$O(2^{\mathrm{poly}(1/\epsilon)})$ since $\theta$ and $\eta$ each
depend polynomially on $1/\epsilon$.  There are $O(n)$ vertices of the
recursion tree and so the time required for the dynamic program, not
including the base cases is $O(2^{\mathrm{poly}(1/\epsilon)}n)$.

The total running time of the dynamic program is
$O(2^{\mathrm{poly}(1/\epsilon)}n)$ where $n$ is the number of
vertices in the parcel.

\subsection{Correctness}\label{sec:approximation-scheme-correctness}
  
The connecting property guarantees that the final solution is feasible
(satisfying the conectivity requirements).  The definitions of
compatible and consistent guarantee the inductive hypothesis.

We show that the procedure {\sc fill} correctly computes the cost of a
minimum-cost subgraph $H_{u}$ of $W(u)$ that meets the configuration
$K_0$.  We have shown that this is true for the leaves of the
recursion tree.  Since $K$ is the configuration corresponding to the
cut $\delta(\set{u_0})$, $K$ is a star.  Therefore $\cost(K)$ is the
cost of the edges of $\delta(\set{u_0})$: $K$ is both the
configuration and a minimum-cost subgraph that meets that
configuration. Further, $\cost(K \cap(\cup_{i=1}^s K_i))$ is the cost
of the edges of $K$ that are in $K_i$ (for $i = 1, \ldots, s$).
$w(\cap_{i=1}^{s} K_i)$ is equal to the cost of the edges common to
$K_1$ and $K_2$ if $s=2$ and zero otherwise.  By the inductive
hypothesis the cost computed is that of a $H_u$: the subgraph of
$W(u)$ of a minimum-cost graph that meets this configuration.
  
Consider the entries of $DP_r$ where $r$ is the root of $\widehat T$.
Since $\delta(W(r))$ is empty, there is only one
configuration corresponding to this subproblem: the trivial
configuration.  Therefore, the dynamic program finds the optimal
solution in $\BC(P)$.

As argued in Section~\ref{sec:ptas}, combining parcel solutions forms
a valid solution in our input graph.  We need to compare the cost of
the output to the cost of an optimal solution.

Recall that new terminals are added at the parcel
boundaries to guarantee connectivity between the parcels; let
$\rvec^+$ denote the requirements including these new terminals.  Let
$S(G,\rvec)$ denote the optimal solution in graph $G$ with
requirements $\rvec$.

For each parcel $P$, there is a (possibly empty) solution $S_P$ in
$\BC(P)$ for the original and new terminals in $P$ consisting of edges
of $S(\PRG,\rvec) \cup \partial {\cal H}$ (where $\cal H$ is the set
of parcels and $\partial {\cal H}$ is the set of boundary edges of all
parcels).  We
have:
\[ \cost(S(\PRG,\rvec)\cap \BC(P)) \leq \cost(S_P)=\cost(S_P\setminus
\partial {\cal H})+\cost(S_P\cap\partial {\cal H} ).\] Every edge of
$S(\PRG,\rvec)$ not in $\partial {\cal H}$ appears in $S_P$ for
exactly one parcel $P$, and so 
\[\sum_{P \in {\cal H}}
\cost(S_P\setminus
  \partial {\cal H}) \leq \cost(S(\PRG,\rvec)).\]
 Every edge of
$\partial {\cal H}$ appears in at most two parcels, and so 
\[\sum_{P \in {\cal
    H}} \cost(S_P\cap \partial {\cal H}) \leq 2\cdot
\cost(\partial {\cal H}).\]
Since a feasible solution for the original and new terminals in $\PRG$
can be obtained by adding a subset of the edges of $\partial{\cal H}$
to $S(\PRG,\rvec)$, the cost of the output of our algorithms is at most
$$\cost(\partial {\cal H})+\sum_{P \in {\cal H}} {S(\BC(P), \rvec^+)} \leq \cost(S(\PRG, \rvec)) + 3 \cost(\partial {\cal
    H}).$$
Combining the cost of the parcel boundaries, the definition of
$\eta$, and the cost of the mortar graph, we obtain
$\cost(\partial {\cal H}) \leq \frac{1}{2} \epsilon \cost(S(G, \rvec))
= \frac{1}{2} \epsilon \opt$.
Finally, by Theorem~\ref{thm:structure}, the cost of the output is
at most $(1+c\epsilon)\ \opt$.  This gives:
\begin{theorem}
  There is an approximation scheme for solving the $\set{0,1,2}$-edge
  connectivity problem (allowing duplication of edges) in planar
  graphs.  The running time is $O(2^{\mathrm{poly}(1/\epsilon)}n + n
  \log n)$.
\end{theorem}

\paragraph{Comments}
The PTAS framework used is potentially applicable to problems where
(i) the input consists of a planar graph $G$ with edge-costs and a
subset $Q$ of the vertices of $G$ (we call $Q$ the set of {\em
  terminals}), and where (ii) the output spans the terminals.  Steiner
tree and two-edge-connectivity have been solved using this framework.
The PTAS for the subset tour problem~\cite{Klein06} (which was the
inspiration for this framework) can be reframed using this technique.
Since the extended abstract of this work first appeared, Borradaile,
Demaine and Tazari have also this framework to give PTASes for the
same set of problems in graphs of bounded genus~\cite{BDT12}, Bateni,
Hajiaghayi and Marx~\cite{BHM09} have extended the framework to the
Steiner forest problem and Bateni~et~al.~\cite{BCEHKM11} have extended
the framework to prize collecting problems.

\paragraph{Acknowledgements} The authors thank David Pritchard for
comments on early versions of this work and discussions with Baigong Zheng regarding Theorems~\ref{thm:2vc-terminal} and ~\ref{thm:cycle-and-path}.
  This material is based upon
work supported by the National Science Foundation under Grant Nos.~CCF-0964037, CCF-0963921, CCF-14-09520 and by a Natural Science and Research
Council of Canada Postdoctoral Fellowship.

\end{document}